\documentclass[draftcls,onecolumn,12pt]{IEEEtran}%
\usepackage{amsfonts}
\usepackage{graphicx,subcaption}
\usepackage{color}
\usepackage{amsmath,amsfonts,amssymb,amsthm,epsfig,epstopdf,url,array}
\usepackage{url,textcomp}
\usepackage{authblk}
\usepackage{cite}
\newcommand{\bs}{\boldsymbol}
\newtheorem{theorem}{Theorem}

\newtheorem{remark}{Remark}
\makeatletter

\newcommand{\Rmnum}[1]{\expandafter\@slowromancap\romannumeral #1@}
\makeatother
\begin{document}
\title{Cooperative HARQ Assisted NOMA Scheme in Large-scale D2D Networks}
\author{
Zheng~Shi,
Shaodan~Ma,
Hesham~ElSawy,
Guanghua~Yang, \\
and Mohamed-Slim Alouini
\thanks{Zheng Shi and Shaodan Ma are with the Department of Electrical and Computer Engineering, University of Macau, Macao (e-mails:shizheng0124@gmail.com, shaodanma@umac.mo).}
\thanks{Hesham ElSawy and Mohamed-Slim Alouini are with the Computer, Electrical, and Mathematical Science and Engineering (CEMSE) Division, King Abdullah University of Science and Technology (KAUST) Thuwal, Makkah Province, Saudi Arabia (e-mail:hesham.elsawy@kaust.edu.sa, slim.alouini@kaust.edu.sa).}
\thanks{Guanghua Yang is with the Institute of Physical Internet, Jinan University (Zhuhai Campus), Zhuhai, China (e-mail:ghyang@jnu.edu.cn).}
}
\maketitle
\begin{abstract}

This paper develops an interference aware design for cooperative hybrid automatic repeat request (HARQ) assisted non-orthogonal multiple access (NOMA) scheme for large-scale device-to-device (D2D) networks. Specifically, interference aware rate selection and power allocation are considered to maximize long term average throughput (LTAT) and area spectral efficiency (ASE). The design framework is based on stochastic geometry that jointly accounts for the spatial interference correlation at the NOMA receivers as well as the temporal interference correlation across HARQ transmissions. It is found that ignoring the effect of the aggregate interference, or overlooking the spatial and temporal correlation in interference, highly overestimates the NOMA performance and produces misleading design insights.  An interference oblivious selection for the power and/or transmission rates leads to violating the network outage constraints. To this end, the results demonstrate the effectiveness of NOMA transmission and manifest the importance of the cooperative HARQ to combat the negative effect of the network aggregate interference. For instance, comparing to the non-cooperative HARQ assisted NOMA, the proposed scheme can yield an outage probability reduction by $32$\%. Furthermore, an interference aware optimal design that maximizes the LTAT given outage constraints leads to $47$\% throughput improvement over HARQ-assisted orthogonal multiple access (OMA) scheme.

\end{abstract}
\IEEEpeerreviewmaketitle
\hyphenation{HARQ}
\section{Introduction}\label{sec:int}
\subsection{Motivation and literature review}
\IEEEPARstart{T}{he} fifth generation of cellular networks are not only envisioned to enhance the mobile broadband services, but also to support massive number of connections within the Internet-of-Things (IoT) paradigm as well as to provide ultra-reliable low-latency communications for some services\cite{yang2017rapro,gao2015enhanced}. Such new requirements impose unprecedented challenges that cannot be fulfilled via the conventional orthogonal multiple-access (OMA) with centralized base station controlled communications. Instead, the 3GPP considers more aggressive spectral utilization schemes such as device-to-device (D2D) communication~\cite{tehrani2014device,asadi2014survey} and non-orthogonal multiple access (NOMA)~\cite{dai2015non,ding2015application} to support such massive number of connections. Despite the increased interference level imposed by D2D communications, it has been shown that D2D can significantly improve the overall network spatial spectral utilization~\cite{lin_D2D,elsawy1,model2016ali,afshang2015fundamentals}. Thanks to the low-power short range direct proximity transmissions enabled by D2D communication. The NOMA further improves the spectrum utilization by simultaneous transmission from the same source to multiple devices on the same time-frequency resources~\cite{islam2016power,dai2015non}. Specifically, NOMA leverages superposition coding (SC) along with successive interference cancellation (SIC) and multi-user diversity to efficiently enhance spectrum utilization. By allocating more transmission power to the user with poorer channel condition, NOMA can achieve a balanced tradeoff between system throughput and user fairness~\cite{cui2016novel,islam2016power,yang2016general,timotheou2015fairness}.

The foreseen gains of NOMA transmission have triggered several research efforts to optimize its operation. For instance, different power allocation strategies for NOMA transmission are developed in~\cite{cui2016novel, yang2016general, islam2016power}. The work in \cite{ding2016impact} investigates the effect of user pairing on the NOMA sum rate performance. Sub-optimal joint power allocation and user pairing strategy is advocated in~\cite{hina2}. For MIMO networks, the authors in \cite{hanif2016minorization} develop an optimized downlink  procedure to maximize NOMA sum rate under per-user rate constraint. The fairness of NOMA transmission is investigated in~\cite{timotheou2015fairness}. Improving NOMA transmission reliability  via cooperation is studied in \cite{ding2015cooperative} and via hybrid automatic repeat request (HARQ) is studied~\cite{choi2016harq, li2015investigation}. The potential gains from integrating NOMA with D2D communication has been investigated in \cite{zhang2016full}. However, none of \cite{ding2015cooperative, hina2, hanif2016minorization, cui2016novel, yang2016general, islam2016power, ding2016impact, timotheou2015fairness, choi2016harq, li2015investigation, zhang2016full} account for the network aggregate interference, which is significant in current cellular networks specially with D2D communication. Note that such co-channel interference affects the power allocation and rate adaptation, which are very crucial for NOMA transmission. The operation of NOMA under aggregate network interference in uplink cellular networks is studied in \cite{hina1}. However, the model in \cite{hina1} neither accounts for HARQ nor for cooperation, which are fundamentals for reliable NOMA communication. 
\subsection{Contribution }
To the best of the authors' knowledge, this paper is the first to study cooperative HARQ-assisted NOMA in large-scale D2D networks. Using stochastic geometry~\cite{haenggi2012stochastic, tut_h}, we develop a novel mathematical paradigm for cooperative HARQ-assisted NOMA that accounts for the spatial interference correlations among the NOMA receivers as well as the temporal interference correlation across the HARQ transmissions\footnote{Considering the spatial and temporal interference correlations highly complicate the analysis and lead to involved performance expressions. However, it is mandatory to reveal the true network performance and alleviate misleading design insights as shown in \cite{Haenggi_corr1, Haenggi_corr2, Haenggi_corr3, Sawy_corr1, Sawy_corr2, Haenggi_corr4, Haenggi_corr5, corr3, crismani2015cooperative, tanbourgi2014effect} and will be shown in this paper.}. Specifically, we consider a single source two users NOMA scheme and model the interfering D2D devices via a Poisson point process (PPP) (cf. Fig.~\ref{fig:sys_mod}), which is widely accepted for modeling D2D devices~\cite{tut_h,lin_D2D,elsawy1,model2016ali,afshang2015fundamentals}. Exact expressions for the outage probability and long term average throughput (LTAT) are calculated for the two users. Furthermore, simplified approximation for LTAT are proposed and validated via simulations. The approximate expressions are utilized to develop an interference aware rate selection and power allocation for cooperative HARQ-assisted NOMA that maximize different network objectives such as LTAT and area spectral efficiency (ASE). The results show the significance of interference spatial and temporal correlation on the NOMA performance. Further, the gains of NOMA over conventional orthogonal multiple access as well as the gains due to cooperation and HARQ are quantified. The contributions of the paper can be summaried in the following points:
\begin{itemize}
\item The paper develops a novel mathematical model based on stochastic geometry for HARQ assisted cooperative NOMA transmission. The developed mathematical model involves exact as well as accurate approximate expressions for the LTAT and outage probability under spatial and temporal interference correlation. The approximations are advocated to alleviate the computational complexity of the exact expressions and enable optimal network design.
\item The developed mathematical model captures the interwoven decoding performance among the two NOMA receivers due to the spatial interference correlation.
\item The developed mathematical model captures the temporal diversity loss in the HARQ retransmissions due to temporal interference correlation.
\item The numerical results quantify the gain of HARQ as well as the gain of cooperation on the NOMA performance in terms of outage probability and LTAT. The effect of the number of HARQ retransmission is also discussed.
\item The developed mathematical model is utilized to formulate an interference aware design that maximizes different network objectives, such as LTAT and ASE, under outage probability constraints.
\item The results show that an interference-oblivious or a correlation-oblivious design is unable to provide satisfied requirement of outage probabilities.
\item The results show the superiority of the proposed   HARQ assisted cooperative NOMA over the conventional OMA scheme.
\end{itemize}

\subsection{Notation and organization}

Throughout the paper, $\mathbb P[X]$ denotes the probability of an event $X$, $\mathbb E$ refers to the expectation operator, $X \cup Y$ and $X  \cap  Y$ denote the union and the intersection of events $X$ and $Y$, respectively, $[\cdot]^+$ denotes the projection onto the nonnegative orthant, $\left\| \cdot \right\|$ stands for Euclidean norm operation, $\Omega$ denotes the sample space and $\emptyset$ denotes the empty set.

The remainder of this paper is organized as follows. Section \ref{sec:sys_mod} presents the considered cooperative HARQ assisted NOMA scheme for D2D networks along with the underlying assumption. Section \ref{sec:per_ana} then analyzes the performance of the proposed scheme, particularly the LTAT and outage probability. Numerical results are presented for verification and discussion in Section \ref{sec:num}. Finally, Section \ref{sec:cond} concludes this paper.

\section{System Model}\label{sec:sys_mod}

This paper considers a D2D communication network modeled as a homogeneous PPP $\Phi \in \mathbb{R}^2$ with intensity $\lambda$. All D2D devices have backlogged buffers and are always transmitting over a shared frequency channel, which is dedicated to D2D communication. Without loss of generality, we focus on a typical source D2D device that is serving two nearby users, as shown in Fig. \ref{fig:sys_mod}, via cooperative HARQ assisted NOMA scheme shown in Fig. \ref{fig:TX_mod}. Let $z$ be the location of the source device, then the distance between the source device and user $i$ (the user at $o_i$) is denoted by $d_i = \left\| z - o_i \right\|$, where $i \in \{1,2\}$. Exploiting the stationarity of the PPP, we assume that one of the users is located at $o_1=(0,0)$ and the other user is located at $o_2=(D,0)$. Since NOMA protocol takes advantage of the difference between fading channels compared to time division multiple access (TDMA)\cite{ding2015cooperative}, we stipulate that user 1 is closer to the source device than user 2, that is, $d_1 < d_2$, without loss of generality.

\begin{figure}
 \centering
  \includegraphics[width=2.7in]{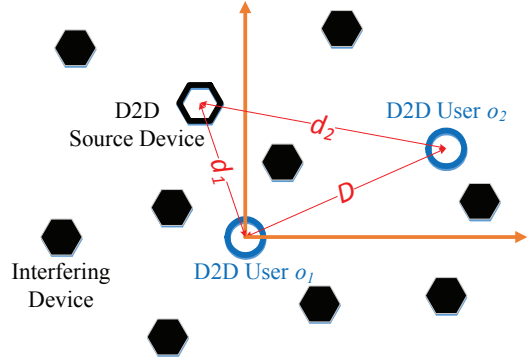}
  \caption{A NOMA-based D2D network model.}\label{fig:sys_mod}
 \end{figure}

As shown in Fig.~\ref{fig:TX_mod}, the cooperative HARQ assisted NOMA transmission is divided into two phases, namely Phase I and Phase II. In Phase I, the source utilizes superposition coding with power domain multiplexing to encode the two signals ${\bf s}_1$ and ${\bf s}_2$ that are intended to the two users 1 and 2, respectively. The nearer user at $o_1$ first decodes ${\bf s}_2$, which is then subtracted via SIC to decode its intended signal ${\bf s}_1$. The farther user 2 directly decodes ${\bf s}_2$ while treating the interfering signal ${\bf s}_1$ as noise, which is denoted hereafter as NOMA interference. The transmission of the superposition messages is repeated until either user 1 or 2 acknowledges successful reception or the maximum number of retransmission $K$ is reached. If either of the two devices acknowledges successful reception, Phase II starts in which the source node only transmits the remaining (i.e., not acknowledged) signal. Furthermore, if user 1 was the  acknowledging receiver, it cooperates with the source and relays ${\bf s}_2$ to user 2. When both users 1 and 2 acknowledge successful reception, the next two signals in the source queue are transmitted via the same aforementioned operation. If the maximum number of transmission $K$ is reached without decoding the intended signals, the signals are dropped from the queue and outage event is declared. For simplicity, we assume that the feedback channel is error-free and delay-free, which can be justified by the low transmission rate and the short length of acknowledgement message.

In this paper, we assume a block Rayleigh fading channel (i.e., channel coefficient remains constant during each HARQ transmission) with known statistical CSI at the source device. However, the channel gain randomly and independently changes from one transmission to another. However, it is important to note that the locations of the interfering devices do not change dramatically over the short HARQ time interval, especially for interferers with low-to-medium mobility. Thus it is reasonable to assume that the interferer locations are fixed during HARQ transmissions, i.e., follow stationary interferer model (SIM) \cite{crismani2015cooperative}, which is valid because of the limited maximal allowable number of transmissions for HARQ in practice, e.g., the maximal number of transmissions is usually chosen up to $5$ and each HARQ round consumes around $8$ms \cite{erceg2001ieee}. The received signal at each of the devices in each transmission phase can be represented as follows


\begin{figure}
 \centering
  \includegraphics[width=5.5in]{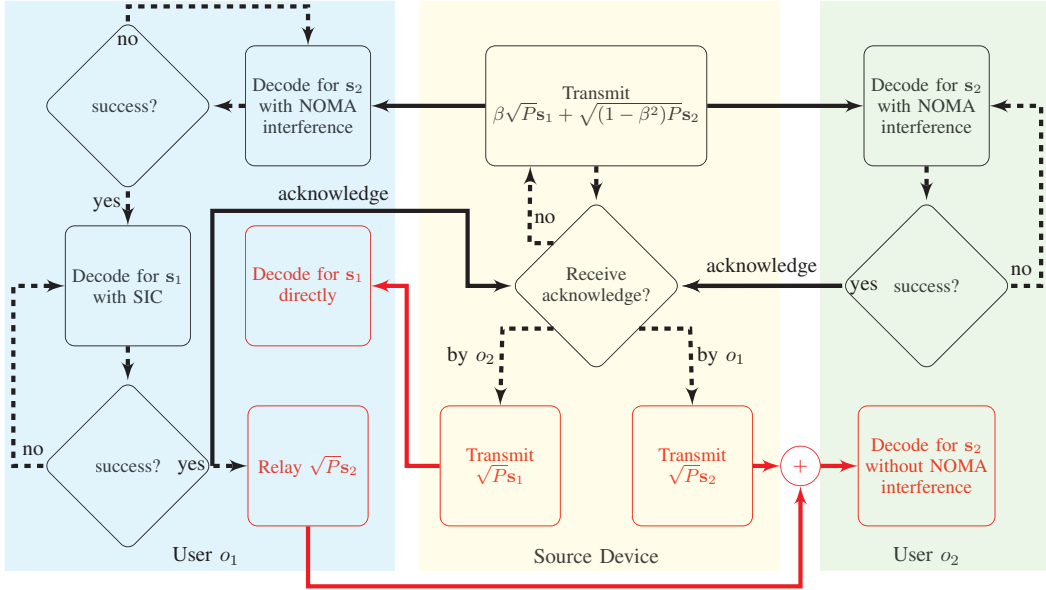}
  \caption{Cooperative HARQ assisted NOMA transmission model for each superposition encoded signal of ${\bf s}_1$ and ${\bf s}_2$, where black color denotes Phase I modes, red color denotes Phase II modes, dotted arrows denote logical state transitions, and solid arrows denote physical transmissions.}\label{fig:TX_mod}
 \end{figure}

\subsubsection{Phase I} The signal received by user $i$ in the $k$-th HARQ round is written as
\begin{align}\label{eqn:signal_received1}
{{\bf{y}}_{i,k}} &= \sqrt {\ell \left( {{d_i}} \right)P} {h_{z{o_i},k}}\left( {{\beta}{{\bf{s}}_1} + \sqrt {1 - {\beta}^2} {{\bf{s}}_2}} \right)  + {\sum \nolimits_{x \in \Phi \backslash \left\{ z \right\}}}\sqrt {\ell \left( {\left\| {x - {o_i}} \right\|} \right)P} {h_{x{o_i},k}}{{\bf{s}}_{x,k}} + {{\bf{n}}_{i,k}},
\end{align}
where $P$ denotes the transmit power and $\beta$ represents the power allocation coefficient. ${\bf s}_i$ is a Gaussian signal with unit variance to user $i$. Each signal of ${{\bf{s}}_1}$ and ${{\bf{s}}_2}$ is encoded independently at the source device and is transmitted with an initial transmission rate $R_i$ for user $i$. ${{\bf{s}}_{x,k}}$ denotes a Gaussian signal with unit variance and delivered by an interfering device located at $x$ in the $k$-th HARQ round. $\ell(d) = d^{-\alpha}$ captures the path loss effect with path loss exponent $\alpha > 2$. The notation ${\Phi }\backslash\left\{ z \right\}$ denotes the set of all devices excluding the source device $z$. ${{\bf{n}}_{i,k}}$ denotes a complex additive white Gaussian noise (AWGN) with zero mean and variance of $\sigma^2$, i.e., ${{\bf{n}}_{i,k}} \sim {\cal CN}(0,\sigma^2)$. ${h_{xo_i,k}}$ denotes the channel coefficient from the interfering device $x$ to user $i$ in the $k$-th HARQ round, and ${h_{zo_i,k}}$ and ${h_{xo_i,k}}$ are complex Gaussian distributed with zero mean and unit variance, i.e., ${h_{zo_i,k}},{h_{xo_i,k}} \sim {\cal CN}(0,1)$.

Following the NOMA protocol, after receiving the signal, the message intended for user 2 is decoded first with SINR
\begin{equation}\label{eqn:sinr_1_2}
{\gamma _{o_i,k,2}^{(I)}} = \frac{{\left( {1 - {\beta}^2} \right)P{{\left| {{h_{z{o_i},k}}} \right|}^2}\ell \left( {{d_i}} \right)}}{{{\beta}^2P{{\left| {{h_{z{o_i},k}}} \right|}^2}\ell \left( {{d_i}} \right) + {I_{i,k}} + {\sigma ^2}}},
\end{equation}
where ${I_{i,k}} $ denotes the total interference at user $i$ from interfering devices $x \in \Phi$ excluding $z$, i.e., $\Phi \backslash \left\{ z \right\}$. More precisely, it follows from (\ref{eqn:signal_received1}) that
\begin{equation}\label{eqn:inter_fer_1}
{I_{i,k}} = P\sum\nolimits_{x \in \Phi \backslash \left\{ z \right\}} {\ell \left( {\left\| {x - {o_i}} \right\|} \right){{\left| {{h_{xo_i,k}}} \right|}^2}}.
\end{equation}
If user 1 successfully decodes the message of user 2, SIC will be carried out to recover its own message ${\bf s}_{1}$ through subtracting the decoded signal ${\bf s}_{2}$ with SINR
\begin{equation}\label{eqn:sinr_1_1}
{\gamma _{o_1,k,1}^{(I)}} = \frac{{{\beta}^2P{{\left| {{h_{z{o_1},k}}} \right|}^2}\ell \left( {{d_1}} \right)}}{{{I_{1,k}} + {\sigma ^2}}}.
\end{equation}
\subsubsection{Phase II}
As shown in Fig.~\ref{fig:TX_mod}, if $s_2$ is successfully decoded prior to the $k$-th HARQ round, the received signal at $o_1$ in the $k$-th HARQ round is therefore given by
\begin{align}\label{eqn:y1_second_phase}
{{\bf y}_{1,k}} &= \sqrt {\ell \left( {{d_1}} \right)P} {h_{z{o_1},k}}{{\bf{s}}_1} + {\sum \nolimits_{x \in \Phi \backslash \left\{ z \right\}}}\sqrt {\ell \left( {\left\| {x - {o_i}} \right\|} \right)P} {h_{x{o_1},k}}{{\bf{s}}_{x,k}}
+ {{\bf{n}}_{1,k}}.
\end{align}
Correspondingly, the received SINR can thus be expressed as
\begin{equation}\label{eqn:sinr_phase2_o1}
\gamma _{{o_1},k,1}^{\left( \Rmnum{2} \right)} = \frac{{P{{\left| {{h_{z{o_1},k}}} \right|}^2}\ell \left( {{d_1}} \right)}}{{{I_{1,k}} + {\sigma ^2}}}.
\end{equation}

\noindent Conversely, if $s_1$ is successfully decoded prior to the $k$-th HARQ round, the received signal at user 2 with cooperation from user 1 in the $k$-th HARQ round is therefore given by
\begin{equation}\label{eqn:phaseII_o2}
{{\bf y}_{2,k}} = \sqrt {\ell \left( D \right)P} {h_{{o_1}{o_2},k}}{{\bf{s}}_2}
 +\sqrt {\ell \left( {{d_2}} \right)P} {h_{z{o_2},k}}{{\bf{s}}_2}
  + {\sum \nolimits_{x \in \Phi \backslash \left\{ z \right\}}}\sqrt {\ell \left( {\left\| {x - {o_2}} \right\|} \right)P} {h_{x{o_2},k}}{{\bf{s}}_{x,k}} + {{\bf{n}}_{2,k}},
\end{equation}
where ${h_{{o_1}{o_2},k}}$ denotes the channel coefficient between two users in the $k$-th transmission. Similar to (\ref{eqn:sinr_phase2_o1}), the received SINR of user 2 can be written as
\begin{align}\label{eqn:SINR_second_phase}
\gamma _{{o_2},k,2}^{\left(
\Rmnum{2} \right)} 
=\frac{{P{{\left| {\sqrt {\ell \left( D \right)} {h_{{o_1}{o_2},k}} + \sqrt {\ell \left( {{d_2}} \right)} {h_{z{o_2},k}}} \right|}^2}}}{{{I_{2,k}} + {\sigma ^2}}}
=\frac{{P{{\left| {{h_{eq,k}}} \right|}^2}\left( {\ell \left( D \right) + \ell \left( {{d_2}} \right)} \right)}}{{{I_{2,k}} + {\sigma ^2}}},
\end{align}
where ${h_{eq,k}} \triangleq \sqrt {\frac{{\ell \left( D \right)}}{{\ell \left( D \right) + \ell \left( {{d_2}} \right)}}} {h_{{o_1}{o_2},k}} + \sqrt {\frac{{\ell \left( {{d_2}} \right)}}{{\ell \left( D \right) + \ell \left( {{d_2}} \right)}}} {h_{z{o_2},k}}$ denotes the equivalent channel coefficient in the $k$-th transmission, and follows a complex Gaussian distribution with zero mean and unit variance, i.e., ${h_{eq,k}} \sim {\cal CN}(0,1)$, and ${\ell \left( D \right) + \ell \left( {{d_2}} \right)}$ is the equivalent path loss.


It is important to note that the signals, and hence the SINRs, are spatially and temporally correlated. For instance, in Phase I, the spatial correlation exists between \eqref{eqn:signal_received1} for $i=1$ and $i=2$ and the temporal correlation exists in \eqref{eqn:signal_received1} across different $k$. In Phase II, only temporal correlations across different $k$ in the signals \eqref{eqn:y1_second_phase} and \eqref{eqn:phaseII_o2} exist since either user $1$ or $2$ is receiving. Similar notion applies to the SINRs given in  \eqref{eqn:sinr_1_2}, \eqref{eqn:sinr_phase2_o1}, and \eqref{eqn:SINR_second_phase}. The spatial and temporal correlations stem from the fact that the two users see common interfering sources across the HARQ rounds. Consequently, the decoding performed at the two users is interwoven due to the spatial correlations as well as the NOMA protocol and cooperative communications. As will be shown in the next section, such interdependence between the performance of the two users makes the analysis significantly involved.

\section{Analyses of throughput and Outage Probability}\label{sec:per_ana}
LTAT is a widely adopted performance metric to characterize the performance of HARQ system. Here we adopt the model developed in \cite{caire2001throughput} to calculate the LTAT of the NOMA transmission with HARQ in the limit for large subcodeword length $L$. For notational convenience, let $t$ denote the number of slots and $b_{o_i}(t)$ be the number of information bits, which are intended for user $i$ and successfully decoded by user $i$, up to slot $t$. The total LTAT $\eta$ measured in bps/Hz is defined as
\begin{equation}\label{eqn:throughput_def}
 \eta  = \mathop {\lim }\limits_{t \to \infty } \frac{{{b_{{o_1}}}\left( t \right) + {b_{{o_2}}}\left( t \right)}}{{tL}} = \mathop {\lim }\limits_{t \to \infty } \frac{{{R_{{o_1}}}\left( t \right) + {R_{{o_2}}}\left( t \right)}}{t},
\end{equation}
where $R_{o_i}(t) \triangleq b_{o_i}(t)/L$ denotes the corresponding information bits per second per hertz successfully decoded by user $i$. The event that user $i$ stops the transmission of the current message is treated as a recurrent event \cite{zorzi1996use}. The recurrent event occurs with two random rewards $\mathcal R_{o_1}$ and $\mathcal R_{o_2}$ gained by the two users at $o_1$ and $o_2$, respectively. Thus by using renewal-reward theorem, the LTAT of the cooperative HARQ assisted NOMA system is given by
\begin{equation}\label{eqn:thr_rew}
\eta  = \frac{{\mathbb E\left( \mathcal R_{o_1} \right)+\mathbb E\left( \mathcal R_{o_2} \right)}}{{\mathbb E\left( \mathcal T \right)}}, \quad \rm with~ probability~ 1,
\end{equation}
where $\mathcal T$ is the random number of transmissions between two consecutive occurrences of the recurrent event (inter-renewal time). Note that $\mathcal R_{o_i}=R_i$ bps/Hz if user $i$ successfully recovers its own message, otherwise $\mathcal R_{o_i}=0$ bps/Hz, we have
\begin{align}\label{eqn:average_reward}
\mathbb E\left( {{{\cal R}_{{o_i}}}} \right) 
&=R_{i}(1-{\mathcal{O}_{K,o_i}}),
\end{align}
where ${{\mathcal{O}_{K,o_i}}}$ denotes the outage probability of user $i$ after $K$ HARQ rounds. Moreover, $\mathcal T$ is a discrete random variable with the sample space $\{1,2,\cdots,K\}$ and obeys the probability distribution as
\begin{align}\label{eqn:rand_T_dis}
\mathbb{P} \left[ {\mathcal T = \kappa } \right] =
\left\{ {\begin{array}{*{20}{l}}\begin{array}{l}
\mathcal{O}_{\kappa-1,o_1|o_2} - \mathcal{O}_{\kappa,o_1|o_2},
\end{array}&{\kappa  < K}\\
{\mathcal{O}_{K-1,o_1|o_2},}&{\kappa  = K}
\end{array}} \right.,
\end{align}
where $\mathcal{O}_{\kappa,o_1|o_2}$ denotes the outage event occurring at either user 1 or user 2 after $\kappa$ transmissions. By using inclusion-exclusion identity, it follows that
\begin{equation}\label{eqn:out_or_inexid}
\mathcal{O}_{\kappa,o_1|o_2} = { {{\mathcal{O}_{\kappa ,{o_1}}} + {\mathcal{O}_{\kappa ,{o_2}}} - {\mathcal{O}_{\kappa ,{o_1},{o_2}}}} },
\end{equation}
where ${{\mathcal{O}_{\kappa,o_1,o_2}}}$ represents the probability that both two users fail to decode their own messages after $\kappa$ HARQ rounds. As such, the average number of transmissions $\mathbb E(\mathcal T)$ is obtained by using (\ref{eqn:rand_T_dis}) and (\ref{eqn:out_or_inexid}) as
\begin{equation}\label{eqn:average_transmissions}
\mathbb E\left( \mathcal T \right) = \sum\limits_{\kappa  = 1}^K {\kappa \mathbb{P} \left[ {\mathcal T = \kappa } \right]}
 = 1 + \sum\limits_{\kappa  = 1}^{K-1} { \left( {{\mathcal{O}_{\kappa ,{o_1}}} + {\mathcal{O}_{\kappa ,{o_2}}} - {\mathcal{O}_{\kappa ,{o_1},{o_2}}}} \right)},
\end{equation}

Accordingly, substituting (\ref{eqn:average_reward}) and (\ref{eqn:average_transmissions})
into (\ref{eqn:thr_rew}) leads to
\begin{equation}\label{eqn:noma_harq_throughput}
\eta = \frac{{{R_1}\left( {1 - {\mathcal{O}_{K,o_1}}} \right) + {R_2}\left( {1 - {\mathcal{O}_{K,o_2}}} \right)}}{{1+\sum\nolimits_{\kappa = 1}^{K - 1} {\left( {{\mathcal{O}_{\kappa,o_1}} + {\mathcal{O}_{\kappa,o_2}} - {\mathcal{O}_{\kappa,o_1,o_2}}} \right)} }}.
\end{equation}
Thus the LTAT is expressed as a function of outage probabilities, which are the fundamental performance metrics. It is worth noting that (\ref{eqn:noma_harq_throughput}) is a general expression to evaluate the LTAT of HARQ assisted NOMA system, which is applicable to both cooperative and non-cooperative cases. Following the same analytical approach, it can be readily extended to derive the LTAT of HARQ assisted NOMA system with two more users. To avoid tedious mathematical derivations, we skip the detailed discussion. To proceed with our analysis, the outage probabilities ${{\mathcal{O}_{K,o_1}}}$, ${{\mathcal{O}_{K,o_2}}}$ and ${{\mathcal{O}_{K,o_1,o_2}}}$ are individually derived as follows.
\subsection{The outage event ${{\mathcal{O}_{K,o_1}}}$}
According to the system model in Section \ref{sec:sys_mod}, the decoding performance of ${\bf s}_1$ depends on the number of HARQ rounds consumed by user 1 to successfully decode and subtract ${\bf s}_2$ as well as the number of HARQ rounds consumed by user 2 to decode ${\bf s}_2$. This is because the source transmission power is totally allocated to ${\bf s}_1$ after user 2 acknowledges successful decoding.  To facilitate our analysis, we define the following events.
\begin{description}
  \item[${{\Theta _{{o_1},i,l}}}$]: The event that user 1 successfully decodes the signal ${\bf s}_i$ after $l$ HARQ rounds;
  \item[$ {{\bar \Theta _{{o_1},i}}}$]: The complement of the union ${\bigcup\limits_{l = 1}^K {{\Theta _{{o_1},i,l}} } }$, that is, user 1 fails to decode the signal ${\bf s}_i$ after $K$ HARQ rounds;
  \item[${\Theta _{{o_2,k}}}$]: The event that user 2 succeeds in decoding its own message after $k$ HARQ rounds;
  \item[${\bar \Theta _{{o_2}}}$]: The complement of the union ${\bigcup\limits_{k = 1}^K {{\Theta _{{o_2},k}} } }$, that is, user 2 fails to recover its own message after $K$ HARQ rounds.
\end{description}
With the above definitions, the outage probability of user 1, i.e., ${\mathcal{O}_{K,{o_1}}}$, can be obtained by using law of total probability as
\begin{align}\label{eqn:p_K_o1_tot}
{\mathcal{O}_{K,{o_1}}} &= \mathbb P\left[ {{{\bar \Theta }_{{o_1},1}}} \right] = \mathbb P\left[ {{{\bar \Theta }_{{o_1},1}}},\Omega,\Omega \right] \notag \\
 &= \mathbb P\left[ {{{\bar \Theta }_{{o_1},1}},\left( {\bigcup\limits_{l = 1}^K {{\Theta _{{o_1},2,l}} } } \right)\bigcup {{{\bar \Theta }_{{o_1},2}}} ,\left( {\bigcup\limits_{k = 1}^K {{\Theta _{{o_2},k}} } } \right)\bigcup {{{\bar \Theta }_{{o_2}}}} } \right].
\end{align}
Notice that ${{\Theta _{{o_1},2,1}},\cdots,{\Theta _{{o_1},2,K}}}$ and ${{{\bar \Theta }_{{o_1},2}}}$ are mutually exclusive events. That is, the intersection of any sequence of these events is empty. Similarly, ${{\Theta _{{o_2},1}} ,\cdots,{\Theta _{{o_2},K}}}$ and ${{{\bar \Theta }_{{o_2}}}}$ are also mutually exclusive. In addition, ${\Theta _{{o_1},2,l}}$ and ${\Theta _{{o_2,k}}} $ are mutually exclusive if $l > k$, since the source only sends ${\bf s}_1$ after the acknowledgement of user 2, and hence, SIC in not required.  Therefore, (\ref{eqn:p_K_o1_tot}) can be simplified as
\begin{align}\label{eqn:out_K_o1}
{\mathcal{O}_{K,{o_1}}} &= \sum\limits_{k = 1}^K {\sum\limits_{l = 1}^k {\mathbb P\left[ {{{\bar \Theta }_{{o_1},1}},{\Theta _{{o_1},2,l}},{\Theta _{{o_2},k}}} \right]} } + \sum\limits_{k = 1}^K {\mathbb P\left[ {{{\bar \Theta }_{{o_1},1}},{{\bar \Theta }_{{o_1},2}},{\Theta _{{o_2},k}} } \right]} \notag\\
 &\quad + \sum\limits_{l = 1}^K {\mathbb P\left[ {{{\bar \Theta }_{{o_1},1}},{\Theta _{{o_1},2,l}},{{\bar \Theta }_{{o_2}}}} \right]} + \mathbb P \left[ {{{\bar \Theta }_{{o_1},1}},{{\bar \Theta }_{{o_1},2}},{{\bar \Theta }_{{o_2}}}} \right].
\end{align}
The terms at the right hand side of (\ref{eqn:out_K_o1}) will be derived one by one as follows.
\subsubsection{${\mathbb P\left[ {{{\bar \Theta }_{{o_1},1}},{\Theta _{{o_1},2,l}} ,{\Theta _{{o_2},k}}} \right]}$}
From information-theoretical perspective, an outage event happens when the mutual information is less than the transmission rate. Herein, ${\mathbb P\left[ {{{\bar \Theta }_{{o_1},1}},{\Theta _{{o_1},2,l}} ,{\Theta _{{o_2},k}}} \right]}$ represents the outage probability of user 1 after SIC given that decoding ${\bf s}_2$ by user 1 consumed  $l$ HARQ rounds and decoding ${\bf s}_2$ by user 2 consumed $k$ HARQ rounds. Note that $l \le k$ should be satisfied in this case. With the signal model in Section \ref{sec:sys_mod}, ${\mathbb P\left[ {{{\bar \Theta }_{{o_1},1}},{\Theta _{{o_1},2,l}} ,{\Theta _{{o_2},k}}} \right]}$ can be obtained as
\begin{equation}\label{eqn:outage_o1_1}
{\mathbb P\left[ {{{\bar \Theta }_{{o_1},1}},{\Theta _{{o_1},2,l}} ,{\Theta _{{o_2},k}}} \right]} =
\mathbb P\left[ {\begin{array}{*{20}{c}}
{\bigcap\limits_{j = l}^k {\mathcal I\left( {\gamma _{{o_1},j,1}^{(I)}} \right) < {R_1}} ,\bigcap\limits_{j = k + 1}^K {\mathcal I\left( {\gamma _{{o_1},j,1}^{\left( {II} \right)}} \right) < {R_1}} ,\bigcap\limits_{j = 1}^{l - 1} {\mathcal I\left( {\gamma _{{o_1},j,2}^{(I)}} \right) < {R_2}} ,}\\
{\mathcal I\left( {\gamma _{{o_1},l,2}^{(I)}} \right) \ge {R_2},\bigcap\limits_{j = 1}^{k - 1} {\mathcal I\left( {\gamma _{{o_2},j,2}^{(I)}} \right) < {R_2}}, \mathcal I\left( {\gamma _{{o_2},k,2}^{(I)}} \right) \ge {R_2}}
\end{array}} \right],
\end{equation}
where $\mathcal I(\gamma) = {\log _2}\left( {1 + \gamma } \right)$ denotes the mutual information given SINR $\gamma$. It is challenging to derive (\ref{eqn:outage_o1_1}) because of the correlated SINRs, i.e., ${\gamma _{{o_1},j,1}^{(I)}}, {\gamma _{{o_1},j,1}^{\left( {II} \right)}}, {\gamma _{{o_1},j,2}^{(I)}}$ and ${\gamma _{{o_2},j,2}^{(I)}}$, whose correlations stem from the temporally and spatially correlated interference, as pointed out in Section \ref{sec:sys_mod}. Thanks to the tractability provided by stochastic geometry, (\ref{eqn:outage_o1_1}) can be derived in closed-form in Appendix \ref{app:outage_o1_1_f} as
\begin{multline}\label{eqn:out_in_ex_manifin}
{\mathbb P\left[ {{{\bar \Theta }_{{o_1},1}},{\Theta _{{o_1},2,l}} ,{\Theta _{{o_2},k}}} \right]} =\\
 {\left[ \begin{array}{l}
\sum\limits_{{\tau _1} = 0}^{l - 1} {\sum\limits_{{\tau _2} = 0}^{k - l} {\sum\limits_{{\tau _3} = 0}^{K - k} {\sum\limits_{{\tau _4} = 0}^{k - 1} {{{\left( { - 1} \right)}^{\sum\limits_{j = 1}^4 {{\tau _j}} }}C_{l - 1}^{{\tau _1}}C_{k - l}^{{\tau _2}}C_{K - k}^{{\tau _3}}C_{k - 1}^{{\tau _4}}} } } } \times\\
{\left(\Psi \left( {{{\bf{U}}_a},{\bs \tau _a};{\frac{{{2^{{R_2}}} - 1}}{{\left( {1 - {2^{{R_2}}}{\beta ^2}} \right)\ell \left( {{d_2}} \right)}}},{\tau _4} + 1} \right) - \Psi \left( {{{\bf{U}}_a},{\bs \tau _b};{\frac{{{2^{{R_2}}} - 1}}{{\left( {1 - {2^{{R_2}}}{\beta ^2}} \right)\ell \left( {{d_2}} \right)}}},{\tau _4} + 1} \right)\right)}
\end{array} \right]^ + },
\end{multline}
where ${\bs{\tau }}_a = \left( {{\tau _1}+1,{\tau _2},{\tau _3}} \right)$, ${{\bs{\tau }}_b} = \left( {{\tau _1},{\tau _2}{\rm{ + }}1,{\tau _3}} \right)$ and $ {\bf{U}}_a = \left( {\frac{{{2^{{R_2}}} - 1}}{{\left( {1 - {2^{{R_2}}}{\beta ^2}} \right)\ell \left( {{d_1}} \right)}},\frac{{{2^{{R_1}}} - 1}}{{{\beta ^2}\ell \left( {{d_1}} \right)}},\frac{{{2^{{R_1}}} - 1}}{{\ell \left( {{d_1}} \right)}}} \right)$.
Herein, it should be mentioned that ${1 - {2^{{R_2}}}{\beta ^2}} > 0$, otherwise user 1 is unable to mitigate the NOMA interference ${\bf s}_2$. In addition, the function of $\Psi ({\bf{U}},\bs \tau ;\hat {\bf{U}},\hat {\bs \tau} )$ is defined as
\begin{equation}\label{eqn:genr_decoding_scu1}
\Psi ({\bf{U}},\bs \tau ;\hat {\bf{U}},\hat {\bs \tau} ) = e^{{ - \frac{{{\sigma ^2}}}{P}\left( {{{\bf{U}}\bs \tau}^{\rm{T}}  + {{\hat {\bf{U}}}}\hat{\bs \tau}^{\rm{T}} } \right) - \lambda \varphi ({\bf{U}},\bs \tau ;\hat {\bf{U}},\hat{\bs \tau} )}}
,
\end{equation}
where ${\bf U} = (U_1,\cdots,U_N)$, ${\bs \tau} = (\tau_1,\cdots,\tau_N)$, $\hat{\bf U} = (\hat U_1,\cdots,\hat U_M)$, $\hat{\bs \tau} = (\hat \tau_1,\cdots,\hat \tau_M)$, and
\begin{equation}\label{eqn:varphi_given}
\varphi ({\bf{U}},\bs \tau ;\hat {\bf{U}},\hat{\bs \tau} ) =
\int\nolimits_{{\mathbb R^2}} {\left( \begin{array}{l}
1 - \prod\limits_{n = 1}^N {\frac{1}{{{{\left( {1 + {U_n}\ell \left( {\left\| u \right\|} \right)} \right)}^{{\tau _n}}}}}} \prod\limits_{n = 1}^M {\frac{1}{{{{\left( {1 + {{\hat U}_n}\ell \left( {\left\| {u + {o_1} - {o_2}} \right\|} \right)} \right)}^{{{\hat \tau }_n}}}}}}
\end{array} \right)du}.
\end{equation}

\subsubsection{${\mathbb P\left[ {{{\bar \Theta }_{{o_1},1}},{{\bar \Theta }_{{o_1},2}},{\Theta _{{o_2},k}}} \right]}$} Once user 2 succeeds in decoding ${\bf s}_2$ after $k$ HARQ rounds, the source device will deliver only ${\bf s}_1$ in subsequent retransmissions, which will be straightforward decoded at user 1 without the use of SIC.
Accordingly, ${\mathbb P\left[ {{{\bar \Theta }_{{o_1},1}},{{\bar \Theta }_{{o_1},2}},{\Theta _{{o_2},k}}} \right]}$ can be written as
\begin{equation}\label{eqn:outage_o1_2}
{\mathbb P\left[ {{{\bar \Theta }_{{o_1},1}},{{\bar \Theta }_{{o_1},2}},{\Theta _{{o_2},k}}} \right]}
 =
 \mathbb P\left[ {\begin{array}{*{20}{c}}
{\bigcap\limits_{j = k + 1}^K {\mathcal I\left( {\gamma _{{o_1},j,1}^{(\Rmnum{2})}} \right) < {R_1}} ,}
{\bigcap\limits_{j = 1}^k {\mathcal I\left( {\gamma _{{o_1},j,2}^{(I)}} \right) < {R_2}} ,}\\
{\bigcap\limits_{j = 1}^{k - 1} {\mathcal I\left( {\gamma _{{o_2},j,2}^{(I)}} \right) < {R_2}} ,\mathcal I\left( {\gamma _{{o_2},k,2}^{(I)}} \right) \ge {R_2},}
\end{array}} \right].
\end{equation}
With the same approach in Appendix \ref{app:outage_o1_1_f}, (\ref{eqn:outage_o1_2}) can be derived as
\begin{multline}\label{eqnp112out1fin}
{\mathbb P\left[ {{{\bar \Theta }_{{o_1},1}},{{\bar \Theta }_{{o_1},2}},{\Theta _{{o_2},k}}} \right]} = \sum\limits_{{\tau _1} = 0}^{K - k} {\sum\limits_{{\tau _2} = 0}^k {\sum\limits_{{\tau _3} = 0}^{k - 1} {{{\left( { - 1} \right)}^{\sum\limits_{j = 1}^3 {{\tau _j}} }}C_{K - k}^{{\tau _1}}C_k^{{\tau _2}}C_{k - 1}^{{\tau _3}}}\times } } \\
\quad \Psi \left( {\left( {\frac{{{2^{{R_1}}} - 1}}{{\ell \left( {{d_1}} \right)}},\frac{{{2^{{R_2}}} - 1}}{{\left( {1 - {2^{{R_2}}}{\beta ^2}} \right)\ell \left( {{d_1}} \right)}}} \right),\left( {{\tau _1},{\tau _2}} \right);\frac{{{2^{{R_2}}} - 1}}{{\left( {1 - {2^{{R_2}}}{\beta ^2}} \right)\ell \left( {{d_2}} \right)}},{\tau _3} + 1} \right).
\end{multline}

\subsubsection{${\mathbb P\left[ {{{\bar \Theta }_{{o_1},1}},{\Theta _{{o_1},2,l}} ,{{\bar \Theta }_{{o_2}}}} \right]}$}\label{subsec:p3}
When user 1 successfully decodes ${\bf s}_2$ after $l$ HARQ rounds, it means that user 1 can fully eliminate the NOMA interference in the current and subsequent HARQ rounds utilized to decode ${\bf s}_1$, which improves the outage probability. Therefore, ${\mathbb P\left[ {{{\bar \Theta }_{{o_1},1}},{\Theta _{{o_1},2, l}} ,{{\bar \Theta }_{{o_2}}}} \right]}$ can be expressed as
\begin{equation}\label{eqn:outage_o1_3}
{\mathbb P\left[ {{{\bar \Theta }_{{o_1},1}},{\Theta _{{o_1},2,l}} ,{{\bar \Theta }_{{o_2}}}} \right]} =
\mathbb P\left[ {\begin{array}{*{20}{c}}
{\bigcap\limits_{j = l}^K {\mathcal I\left( {\gamma _{{o_1},j,1}^{(I)}} \right) < {R_1}} ,}
\bigcap\limits_{j = 1}^{l - 1} {\mathcal I\left( {\gamma _{{o_1},j,2}^{(I)}} \right) < {R_2}} ,\\\mathcal I\left( {\gamma _{{o_1},l,2}^{(I)}} \right) \ge {R_2},
{\bigcap\limits_{j = 1}^K {\mathcal I\left( {\gamma _{{o_2},j,2}^{(I)}} \right) < {R_2}} }
\end{array}} \right].
\end{equation}
Likewise, (\ref{eqn:outage_o1_3}) can be derived as
\begin{align}\label{eqn:out_o1_3_fin}
&{\mathbb P\left[ {{{\bar \Theta }_{{o_1},1}},{\Theta _{{o_1},2,l}} ,{{\bar \Theta }_{{o_2}}}} \right]} = \notag \\
& \quad \quad \quad {\left[ {\begin{array}{*{20}{l}}
{\sum\limits_{{\tau _1} = 0}^{l - 1} {\sum\limits_{{\tau _2} = 0}^{K - l} {\sum\limits_{{\tau _3} = 0}^K {{{\left( { - 1} \right)}^{\sum\limits_{j = 1}^3 {{\tau _j}} }}C_{l - 1}^{{\tau _1}}C_{K - l}^{{\tau _2}}C_K^{{\tau _3}}} } }  \times }\\
{\left(
\Psi \left( {{{\bf{U}}_b},{\bs \tau _c};\frac{{{2^{{R_2}}} - 1}}{{\left( {1 - {2^{{R_2}}}{\beta ^2}} \right)\ell \left( {{d_2}} \right)}},{\tau _3}} \right)
 - \Psi \left( {{{\bf{U}}_b},{\bs \tau _d};\frac{{{2^{{R_2}}} - 1}}{{\left( {1 - {2^{{R_2}}}{\beta ^2}} \right)\ell \left( {{d_2}} \right)}},{\tau _3}} \right)
\right)}
\end{array}} \right]^ + },
\end{align}
where ${\bs \tau _c} = \left( {{\tau _1}+1,{\tau _2}} \right)$, ${\bs \tau _d} = \left( {{\tau _1},{\tau _2} + 1} \right)$ and ${{\bf{U}}_b} = \left( {\frac{{{2^{{R_2}}} - 1}}{{\left( {1 - {2^{{R_2}}}{\beta ^2}} \right)\ell \left( {{d_1}} \right)}},\frac{{{2^{{R_1}}} - 1}}{{{\beta ^2}\ell \left( {{d_1}} \right)}}} \right)$.

\subsubsection{$\mathbb P \left[ {{{\bar \Theta }_{{o_1},1}},{{\bar \Theta }_{{o_1},2}},{{\bar \Theta }_{{o_2}}}} \right]$}\label{subsec:p4}
If user 1 fails to mitigate the NOMA interference and user 2 fails to decode its own message after $K$ transmissions, it is impossible for user 1 to decode ${\bf s}_1$. Thus $\mathbb P\left[ {{{\bar \Theta }_{{o_1},1}},{{\bar \Theta }_{{o_1},2}},{{\bar \Theta }_{{o_2}}}} \right]$ is expressed as
\begin{equation}\label{eqn:outage_o1_4}
\mathbb P\left[ {{{\bar \Theta }_{{o_1},1}},{{\bar \Theta }_{{o_1},2}},{{\bar \Theta }_{{o_2}}}} \right] =
\mathbb P\left[ {\begin{array}{*{20}{c}}
{\bigcap\limits_{j = 1}^K {\mathcal I\left( {\gamma _{{o_1},j,2}^{(I)}} \right) < {R_2}} },
{\bigcap\limits_{j = 1}^K {\mathcal I\left( {\gamma _{{o_2},j,2}^{(I)}} \right) < {R_2}} }
\end{array}} \right].
\end{equation}


Similarly, (\ref{eqn:outage_o1_4}) can finally be derived as
\begin{equation}\label{eqn:out_o1_4_fin}
\mathbb P\left[ {{{\bar \Theta }_{{o_1},1}},{{\bar \Theta }_{{o_1},2}},{{\bar \Theta }_{{o_2}}}} \right] =\sum\limits_{{\tau _1} = 0}^K {\sum\limits_{{\tau _2} = 0}^K {{{\left( { - 1} \right)}^{\sum\limits_{j = 1}^2 {{\tau _j}} }}C_K^{{\tau _1}}C_K^{{\tau _2}}} }
\Psi \left( {\frac{{{2^{{R_2}}} - 1}}{{\left( {1 - {2^{{R_2}}}{\beta ^2}} \right)\ell \left( {{d_1}} \right)}},{\tau _1};\frac{{{2^{{R_2}}} - 1}}{{\left( {1 - {2^{{R_2}}}{\beta ^2}} \right)\ell \left( {{d_2}} \right)}},{\tau _2}} \right).
\end{equation}

\subsection{The outage event ${{ \mathcal{O}_{K,o_2}}}$}
Similar to (\ref{eqn:out_K_o1}), the probability of outage event at user 2, i.e., ${\mathcal{O}_{K,{o_2}}}$, can be obtained by using law of total probability as
\begin{multline}\label{eqn:pout_o2}
{\mathcal{O}_{K,{o_2}}} = \mathbb P[\bar \Theta_{o_2}]
=\mathbb P
\left[ {\left( {\bigcup\limits_{k = 1}^K {{\Theta _{{o_1},1,k}}} } \right)\bigcup {{{\bar \Theta }_{{o_1},1}}} ,\left( {\bigcup\limits_{l = 1}^K {{\Theta _{{o_1},2,l}}} } \right)\bigcup {{{\bar \Theta }_{{o_1},2}}} ,{{\bar \Theta }_{{o_2}}}} \right] \\
=\sum\limits_{l = 1}^K {\sum\limits_{k = l}^K {\mathbb P\left[ {{\Theta _{{o_1},1,k}} ,{\Theta _{{o_1},2,l}},{{\bar \Theta }_{{o_2}}}} \right]} }
 + \sum\limits_{l = 1}^K {\mathbb P\left[ {{\bar \Theta }_{{o_1},1}},{\Theta _{{o_1},2,l}},{{\bar \Theta }_{{o_2}}} \right]}
 + \mathbb P\left[ {{\bar \Theta }_{{o_1},1}},{{\bar \Theta }_{{o_1},2}},{{{\bar \Theta }_{{o_2}}}} \right],
\end{multline}
where the last step holds because of ${\Theta _{{o_1},1,k}}\bigcap {{\Theta _{{o_1},2,l}}}  = \emptyset$ if $k < l$ and ${\Theta _{{o_1},1,k}}\bigcap {{{\bar \Theta }_{{o_1},2}}} \bigcap {{{\bar \Theta }_{{o_2}}}} =\emptyset$. Noting that ${\mathbb P\left[ {{{\bar \Theta }_{{o_1},1}},{\Theta _{{o_1},2,l}} ,{{\bar \Theta }_{{o_2}}}} \right]}$ and $\mathbb P \left[ {{{\bar \Theta }_{{o_1},1}},{{\bar \Theta }_{{o_1},2}},{{\bar \Theta }_{{o_2}}}} \right]$ have been derived in Sections \ref{subsec:p3} and \ref{subsec:p4}, respectively. Hence, the remaining term  is ${\mathbb P\left[ {{ \Theta }_{{o_1},1,k}},{\Theta _{{o_1},2,l}} ,{{\bar \Theta }_{{o_2}}} \right]}$, which is derived in the sequel.

Suppose that user 1 successfully decodes ${\bf s}_2$ after $l$ HARQ rounds and ${\bf s}_1$ in the $k$-th HARQ round with SIC, where $k \ge l$. Thereupon, user 1 and the source device cooperate to deliver the message to user 2 in the subsequent transmissions. In this case, the outage probability of user 2 after $K$ HARQ rounds, i.e., ${\mathbb P\left[ {{\Theta _{{o_1},1,k}},{\Theta _{{o_1},2,l}},{{\bar \Theta }_{{o_2}}}} \right]}$, is obtained explicitly by considering the two cases of whether $k=l$ or not. Firstly, if $k=l$, it means that user 1 successfully subtracts NOMA interference and decodes ${\bf s}_1$ at the same HARQ round, ${\mathbb P\left[{{\Theta _{{o_1},1,l}},{\Theta _{{o_1},2,l}},{{\bar \Theta }_{{o_2}}}} \right]}$ can thus be derived as
\begin{equation}\label{eqn:out_o2_1_case1}
{\mathbb P\left[{{\Theta _{{o_1},1,l}},{\Theta _{{o_1},2,l}},{{\bar \Theta }_{{o_2}}}} \right]}
 =
 \mathbb P\left[ {\begin{array}{*{20}{c}}
{\bigcap\limits_{j = 1}^l {{\cal I}\left( {\gamma _{{o_2},j,2}^{(I)}} \right) < {R_2}} ,\bigcap\limits_{j = l + 1}^K {{\cal I}\left( {\gamma _{{o_2},j,2}^{\left( {II} \right)}} \right) < {R_2}} ,}\\
{\bigcap\limits_{j = 1}^{l - 1} {{\cal I}\left( {\gamma _{{o_1},j,2}^{(I)}} \right) < {R_2}} ,{\cal I}\left( {\gamma _{{o_1},l,2}^{(I)}} \right) \ge {R_2},}
{{\cal I}\left( {\gamma _{{o_1},l,1}^{(I)}} \right) \ge {R_1}}
\end{array}} \right].
\end{equation}
By applying the method introduced in Appendix \ref{app:outage_o1_1_f}, we have 
\begin{multline}\label{eqn:out_o1_3fin}
{\mathbb P\left[ {{\Theta _{{o_1},1,l}},{\Theta _{{o_1},2,l}},{{\bar \Theta }_{{o_2}}}} \right]} =
\sum\limits_{{\tau _1} = 0}^{l - 1} {\sum\limits_{{\tau _2} = 0}^l {\sum\limits_{{\tau _3} = 0}^{K - l} {{{\left( { - 1} \right)}^{\sum\limits_{j = 1}^3 {{\tau _j}} }}C_{l - 1}^{{\tau _1}}C_l^{{\tau _2}}C_{K - l}^{{\tau _3}}} \times} } \\
 \Psi \left( \begin{array}{l}
\left( {\frac{{{2^{{R_2}}} - 1}}{{\left( {1 - {2^{{R_2}}}{\beta ^2}} \right)\ell \left( {{d_1}} \right)}},\max \left\{ {\frac{{{2^{{R_2}}} - 1}}{{\left( {1 - {2^{{R_2}}}{\beta ^2}} \right)\ell \left( {{d_1}} \right)}},\frac{{{2^{{R_1}}} - 1}}{{{\beta ^2}\ell \left( {{d_1}} \right)}}} \right\}} \right),\left( {{\tau _1},1} \right);\\
\left( {\frac{{{2^{{R_2}}} - 1}}{{\left( {1 - {2^{{R_2}}}{\beta ^2}} \right)\ell \left( {{d_2}} \right)}},\frac{{{2^{{R_2}}} - 1}}{{\ell \left( D \right){ + \ell \left( {{d_2}} \right)}}}} \right),\left( {{\tau _2},{\tau _3}} \right)
\end{array} \right),
\end{multline}

On the other hand, if $k>l$, that is, the events of the successful message decoding and NOMA interference cancellation at user 1 occur in two different HARQ rounds, ${\mathbb P\left[{{\Theta _{{o_1},1,k}},{\Theta _{{o_1},2,l}},{{\bar \Theta }_{{o_2}}}} \right]}$ can be expressed as
\begin{equation}\label{eqn:out_o2_1_case2}
{\mathbb P\left[{{\Theta _{{o_1},1,k}},{\Theta _{{o_1},2,l}},{{\bar \Theta }_{{o_2}}}} \right]}
 =
  \mathbb P\left[ {\begin{array}{*{20}{c}}
{\bigcap\limits_{j = 1}^k {\mathcal I\left( {\gamma _{{o_2},j,2}^{(I)}} \right) < {R_2}} ,\bigcap\limits_{j = k + 1}^K {\mathcal I\left( {\gamma _{{o_2},j,2}^{(\Rmnum{2})}} \right) < {R_2}} , \bigcap\limits_{j = 1}^{l - 1} {\mathcal I\left( {\gamma _{{o_1},j,2}^{(I)}} \right) < {R_2}} ,}\\
{\mathcal I\left( {\gamma _{{o_1},l,2}^{(I)}} \right) \ge {R_2},}
{\bigcap\limits_{j = l}^{k - 1} {\mathcal I\left( {\gamma _{{o_1},j,1}^{(I)}} \right) < {R_1}} ,\mathcal I\left( {\gamma _{{o_1},k,1}^{(I)}} \right) \ge {R_1}}
\end{array}} \right].
\end{equation}
Similarly, ${\mathbb P\left[{{\Theta _{{o_1},1,k}},{\Theta _{{o_1},2,l}},{{\bar \Theta }_{{o_2}}}} \right]}$ can be eventually derived as
\begin{equation}\label{eqn:out_o1_3finkll}
{\mathbb P\left[{{\Theta _{{o_1},1,k}},{\Theta _{{o_1},2,l}},{{\bar \Theta }_{{o_2}}}} \right]} =\\
\left[\begin{array}{l}
\sum\limits_{{\tau _1} = 0}^{l - 1} {\sum\limits_{{\tau _2} = 0}^{k - l - 1} {\sum\limits_{{\tau _3} = 0}^k {\sum\limits_{{\tau _4} = 0}^{K - k} {{{\left( { - 1} \right)}^{\sum\limits_{j = 1}^4 {{\tau _j}} }}C_{l - 1}^{{\tau _1}}C_{k - l - 1}^{{\tau _2}}C_k^{{\tau _3}}C_{K - k}^{{\tau _4}}} } } }\\
\quad \times \left( {\Psi \left( {{{\bf{U}}_d},{\bs \tau _f};{{\bf{U}}_c},{\bs \tau _e}} \right) - \Psi \left( {{{\bf{U}}_d},{\bs \tau _g};{{\bf{U}}_c},{\bs \tau _e}} \right)} \right)
\end{array}\right]^+, \\
k > l,
\end{equation}
where ${{\bs{\tau }}_e} = \left( {{\tau _3},{\tau _4}} \right)$, ${{\bs{\tau }}_f} = \left( {{\tau _1}+1,{\tau _2} + 1} \right)$ and ${{\bs{\tau }}_g} = \left( {{\tau _1},{\tau _2} + 2} \right)$, ${{\bf{U}}_c} = \left( {\frac{{{2^{{R_2}}} - 1}}{{\left( {1 - {2^{{R_2}}}{\beta ^2}} \right)\ell \left( {{d_2}} \right)}},\frac{{{2^{{R_2}}} - 1}}{{\ell \left( D \right){ + \ell \left( {{d_2}} \right)}}}} \right)$ and ${{\bf{U}}_d} = \left( {\frac{{{2^{{R_2}}} - 1}}{{\left( {1 - {2^{{R_2}}}{\beta ^2}} \right)\ell \left( {{d_1}} \right)}},\frac{{{2^{{R_1}}} - 1}}{{{\beta ^2}\ell \left( {{d_1}} \right)}}} \right)$.


\subsection{The outage event ${{\mathcal{O}_{K,o_1,o_2}}}$}
Analogous to (\ref{eqn:out_K_o1}) and (\ref{eqn:pout_o2}), it follows by using law of total probability that
\begin{align}\label{eqn:out_o1_o2_rew}
{\mathbb \mathcal{O}_{K,{o_1},{o_2}}} &= \mathbb P\left[ {{{\bar \Theta }_{{o_1},1}},{{\bar \Theta }_{{o_2}}}} \right]
= \mathbb P\left[ {{{\bar \Theta }_{{o_1},1}},{\left( {\bigcup\limits_{l = 1}^K {{\Theta _{{o_1},2,l}}} } \right)\bigcup {{{\bar \Theta }_{{o_1},2}}} },{{\bar \Theta }_{{o_2}}}} \right]\notag\\
&= \sum\limits_{l = 1}^K {\mathbb P\left[ {{{\bar \Theta }_{{o_1},1}},{\Theta _{{o_1},2,l}},{{\bar \Theta }_{{o_2}}}} \right]}  + \mathbb P\left[ {{\bar \Theta }_{{o_1},1}},{{\bar \Theta }_{{o_1},2}},{{{\bar \Theta }_{{o_2}}}} \right],
\end{align}
where ${\mathbb P\left[ {{{\bar \Theta }_{{o_1},1}},{\Theta _{{o_1},2,l}},{{\bar \Theta }_{{o_2}}}} \right]}$ and $\mathbb P\left[ {{\bar \Theta }_{{o_1},1}},{{\bar \Theta }_{{o_1},2}},{{{\bar \Theta }_{{o_2}}}} \right]$ have been given by (\ref{eqn:out_o1_3_fin}) and (\ref{eqn:out_o1_4_fin}), respectively.

Accordingly, the outage probabilities ${{\mathcal{O}_{K,o_1}}}$, ${{\mathcal{O}_{K,o_2}}}$ and ${{\mathcal{O}_{K,o_1,o_2}}}$ can be calculated by using (\ref{eqn:out_K_o1}), (\ref{eqn:pout_o2}) and (\ref{eqn:out_o1_o2_rew}), respectively. Substituting them into (\ref{eqn:noma_harq_throughput}) yields the LTAT of the proposed scheme. In order to evaluate the outage probabilities, it essentially resorts to the calculation of the double integral of $\varphi ({\bf{U}},\bs \tau ;\hat {\bf{U}},\hat{\bs \tau} )$ in (\ref{eqn:varphi_given}). Unfortunately, the double integral representation of (\ref{eqn:varphi_given}) entails a high computational complexity on the performance evaluation. Alternatively, we propose an accurate approximation approach to compute (\ref{eqn:varphi_given}) effectively. Since it is usually expected that NOMA users are not far away from each other due to the exploitation of cooperative communications, i.e., small $D$, we have the following theorem to obtain an accurate approximation of $\varphi ({\bf{U}},\bs \tau ;\hat {\bf{U}},\hat{\bs \tau} )$.
\begin{theorem}\label{the:app} For small $D$, $\varphi ({\bf{U}},\bs \tau ;\hat {\bf{U}},\hat{\bs \tau} )$ in (\ref{eqn:varphi_given}) can be written as 
\begin{multline}\label{eqn:varphi_0}
\varphi ({\bf{U}},{\bs \tau} ;{\bf{\hat U}},\hat {\bs \tau} ) \approx \varphi (\tilde{\bf{U}},\tilde{\bs \tau};{\bf{0}},{\bf{0}} ) = \pi{\rm{B}}\left( {1 - \frac{2}{\alpha },\sum\limits_{\iota  = 1}^{N+M} {{\tilde \tau _\iota }}  + 1} \right)\sum\limits_{\kappa  = 1}^{N+M} {{{{\tilde \tau _\kappa }{\tilde U_\kappa }{\tilde U_\mu }^{\frac{2}{\alpha } - 1}}}  }\times \\
 F_D^{\left( {{N+M} - 1} \right)}\left( {1 - \frac{2}{\alpha },\left( {{\tilde \tau _\iota } + {\delta _{\iota  - \kappa }}} \right)_{\iota  = 1,\iota  \ne \mu }^{N+M};\sum\limits_{\iota  = 1}^{N+M} {{\tilde \tau _\iota }}  + 1;\left( {1 - \frac{{{\tilde U_\iota }}}{{{\tilde U_\mu }}}} \right)_{\iota  = 1,\iota  \ne \mu }^{N+M}} \right),
\end{multline}
wherein $\tilde{\bf U} = ({\bf U},\hat{\bf U})=(\tilde U_1,\cdots,\tilde U_{N+M})$ and $\tilde{\bs \tau} = ({\bs \tau},\hat{\bs \tau})=(\tilde \tau_1,\cdots,\tilde \tau_{N+M})$, ${{\delta }_s}$ denotes Dirac function, $F_D^{({N})}(\cdot)$ denotes the fourth kind of Lauricella function \cite[Eq. A.52]{mathai2009h} and ${\rm B}(a,b) = \frac{\Gamma(a)\Gamma(b)}{\Gamma(a+b)}$ represents Beta function. 
\end{theorem}
\begin{proof}
  Please see Appendix \ref{app:proof_o1eqo2}.
\end{proof}
It is worth noting that the simple and closed-form expression of $\varphi ({\bf{U}},\bs \tau ;\hat {\bf{U}},\hat{\bs \tau} )$ can significantly facilitate later optimal system design.
%
%

\section{Numerical Results and Discussions}\label{sec:num}
This section first validates the developed mathematical model via independent system level simulations. Numerical results are also presented to demonstrate the effect of interference on NOMA performance as well as to quantify the gains offered by the proposed NOMA scheme. The proposed interference-aware design for the cooperative HARQ-assisted NOMA scheme is then presented. Note that the approximation approach of (\ref{eqn:varphi_0}) in Theorem \ref{the:app} is utilized to optimize system performance, including the maximization of LTAT and the maximization of ASE. Unless otherwise specified, the network parameters are selected as follows: $d_1=5$m, $d_2 = 10$m, $R_1 = 4R_2 = 2$ bps/Hz, $D=10$m, $\beta^2=0.3$ and $\lambda=5*10^{-5}{{\rm m}^{-2}}$.
\subsection{Verification}
In Fig. \ref{fig:ver_ltat}, the LTAT is plotted against the transmit signal-to-noise $\frac{P}{\sigma^2}$ (SNR, the ratio of transmit power to AWGN power) for different $K$, where Monte Carlo simulations are conducted to confirm the analysis. With regard to the approximation approach for $K=1$, it is readily found from (\ref{eqn:varphi_11}) that the approximation in (\ref{eqn:varphi_0}) becomes an equality, and hence, the exact results for $K=1$ can be obtained with (\ref{eqn:varphi_0}). Clearly, Fig. \ref{fig:ver_ltat} shows an excellent agreement between the simulation results and the exact results, and justifies the accuracy of the approximation results as well. Not surprisingly, the LTAT can be improved through increasing the transmit SNR, while it saturates in high SNR regime at a value lower than the sum of transmission rates, i.e., $R_1 + R_2 = 2.5 $bps/Hz, due to the interference incurred by other active D2D transmitters. Additionally, as shown in Fig. \ref{fig:ver_ltat}, we should pay attention to the fact that the increase of the maximal number of transmissions $K$ may yield the deterioration of the LTAT because of the intricate relationship between $\eta$ and $K$. This is essentially due to (\ref{eqn:thr_rew}) that shows that increasing the maximal number of transmissions allows more  information bits to be successfully delivered, nevertheless, the average number of transmissions ${{\mathbb E\left( \mathcal T \right)}}$ increases. The contradictory effects of increasing $K$ thus result in different tendencies of $\eta$ with respect to $K$ under low SNR and under high SNR. More specifically, the increase of $K$ is favorable for $\eta$ below a certain SNR threshold, whereas continuing to increase SNR would become counterproductive for $\eta$.

\begin{figure*}[t!]
    \centering
    \begin{subfigure}[t]{0.45\textwidth}
  \centerline{\includegraphics[width=  3 in]{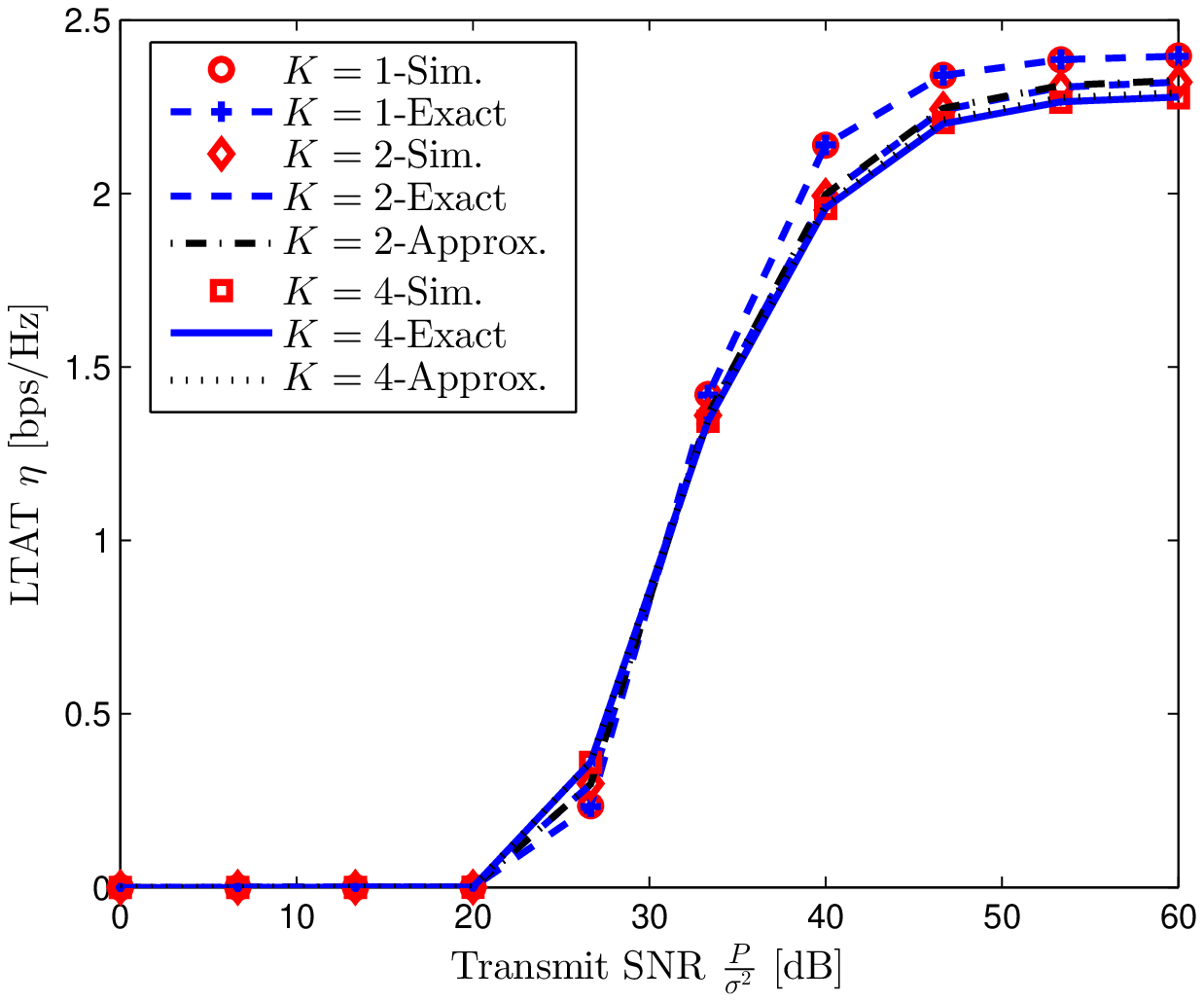}}
      \caption{\, LTAT $\eta$.}
\label{fig:ver_ltat}
   \end{subfigure}
    ~
    \begin{subfigure}[t]{0.45\textwidth}
       \centerline{\includegraphics[width=  3.1 in]{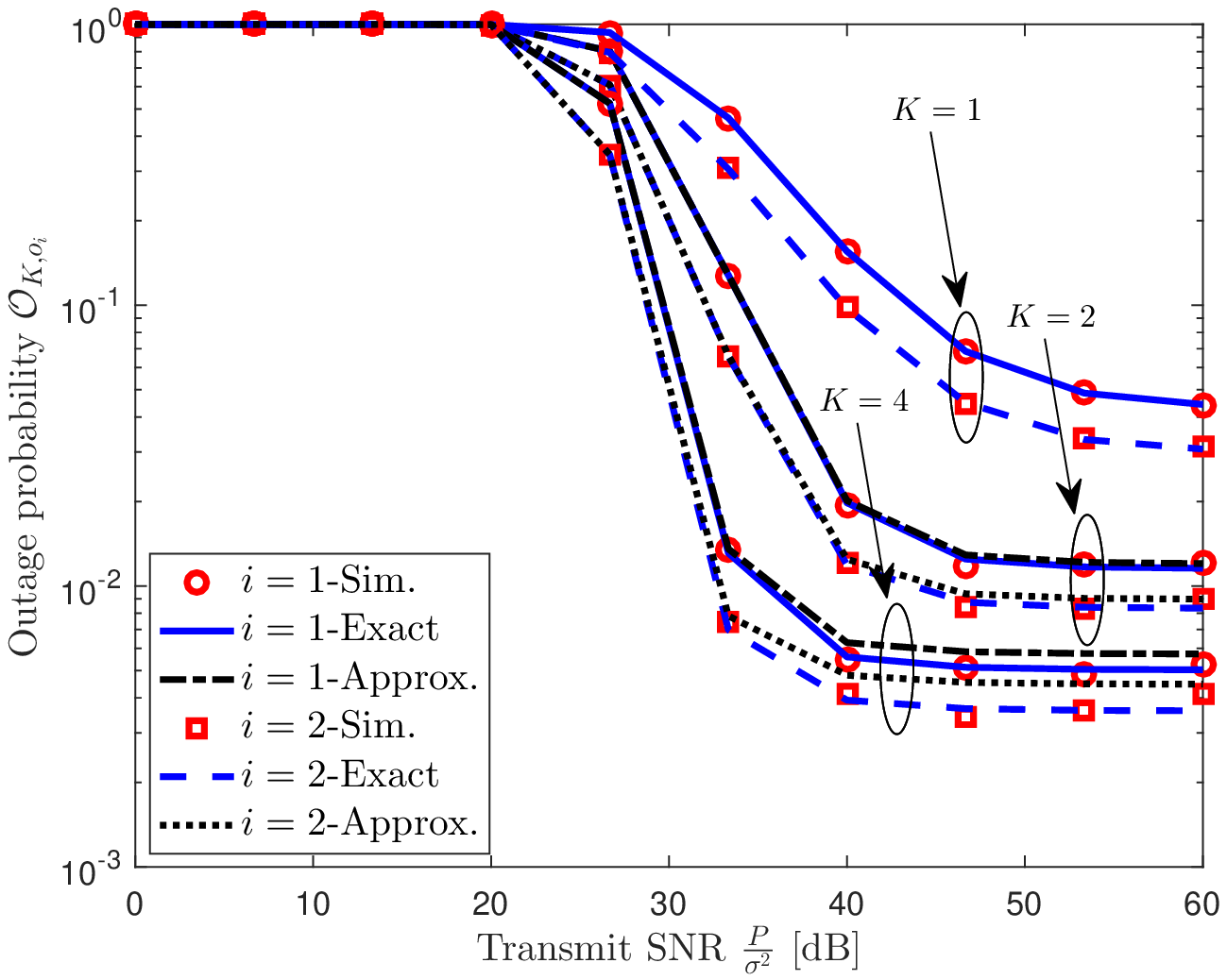}}\caption{\, Outage probability $\mathcal{O}_{K,o_i}$.}
\label{fig:ver_out}
    \end{subfigure}
       \caption{The effect of transmit SNR $\frac{P}{\sigma^2}$.}
       \label{start_fig}
\end{figure*}


Fig. \ref{fig:ver_out} illustrates the outage probabilities of the two NOMA users versus the transmit SNR. The figure further verifies the accuracy of the exact and approximation results. It is easily seen that the outage probabilities of both two users decrease with the transmit SNR but would converge to certain outage floors in high SNR regime due to the co-channel interference, as elucidated in Fig. \ref{fig:ver_ltat}. Moreover, unlike Fig. \ref{fig:ver_ltat}, the outage probabilities can be significantly reduced through increasing the maximal number of transmissions, which manifests  the improved  reliability offered by HARQ. 


\subsection{Effect of Spatially and Temporally Correlated Interference}
Considering temporal and spatial correlation among interferences across all HARQ rounds is important to reveal the true system performance. To illustrate the adverse impact of spatially and temporally correlated interference, Figs. \ref{fig:sptepcorr_ltat} and \ref{fig:sptepcorr_out} compare, respectively, the LTAT $\eta$ and outage probability $\mathcal{O}_{K,o_i}$ of the considered correlated interference model with those of two other simpler interference models that i) ignore the effect of co-channel interferences (labeled as ``No Inter.'' in figures), and ii) ignore the spatial and temporal correlation in co-channel interferences (labeled as ``No Corr.'' in figures). As shown in the two figures, the two simpler interference models provide an unrealistic overestimate of the NOMA performance compared to actual performance especially in high SNR regime. For instance, Fig. \ref{fig:sptepcorr_ltat} shows that the models ignoring interference correlation and assuming no interference overestimate the true performance of LTAT by up to 2\% and 10\%, respectively. Fig. \ref{fig:sptepcorr_out} shows that the actual outage probability is considerably higher than the two simpler interference models by roughly $10^3$\texttildelow$10^4$ times for a fixed value of transmit SNR $\frac{P}{\sigma^2}=40$dB. This is because the temporal and spatial correlation in interferences captures the diversity losses due to the fixed interferers locations~\cite{tanbourgi2014effect}. Therefore, accounting for the spatial and temporal correlation is mandatory to reveal the true system performance.


\begin{figure*}[t!]
    \centering
    \begin{subfigure}[t]{0.45\textwidth}
  \centerline{\includegraphics[width=  3 in]{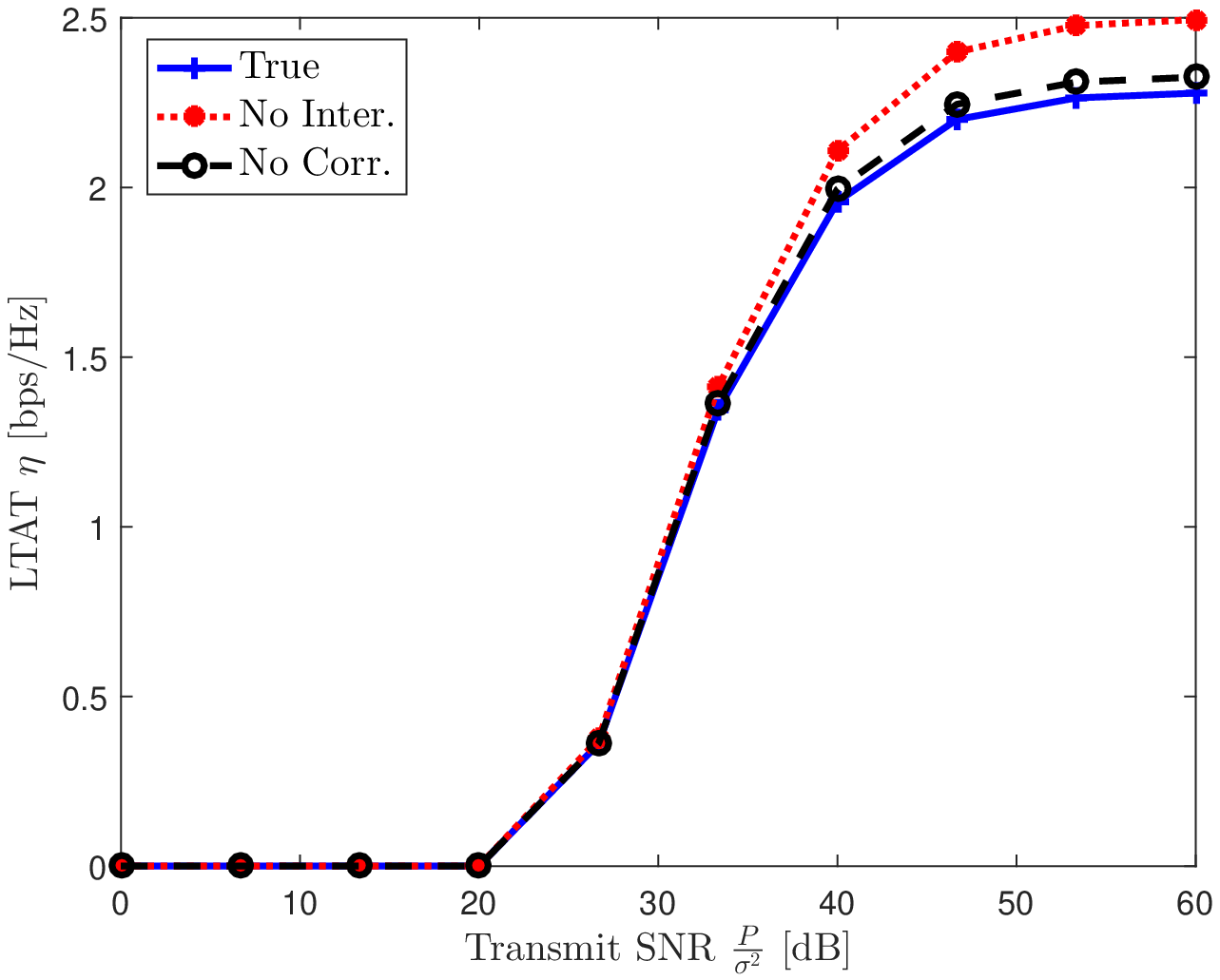}}
      \caption{\, LTAT $\eta$.}
\label{fig:sptepcorr_ltat}
   \end{subfigure}
    ~
    \begin{subfigure}[t]{0.45\textwidth}
       \centerline{\includegraphics[width=  3 in]{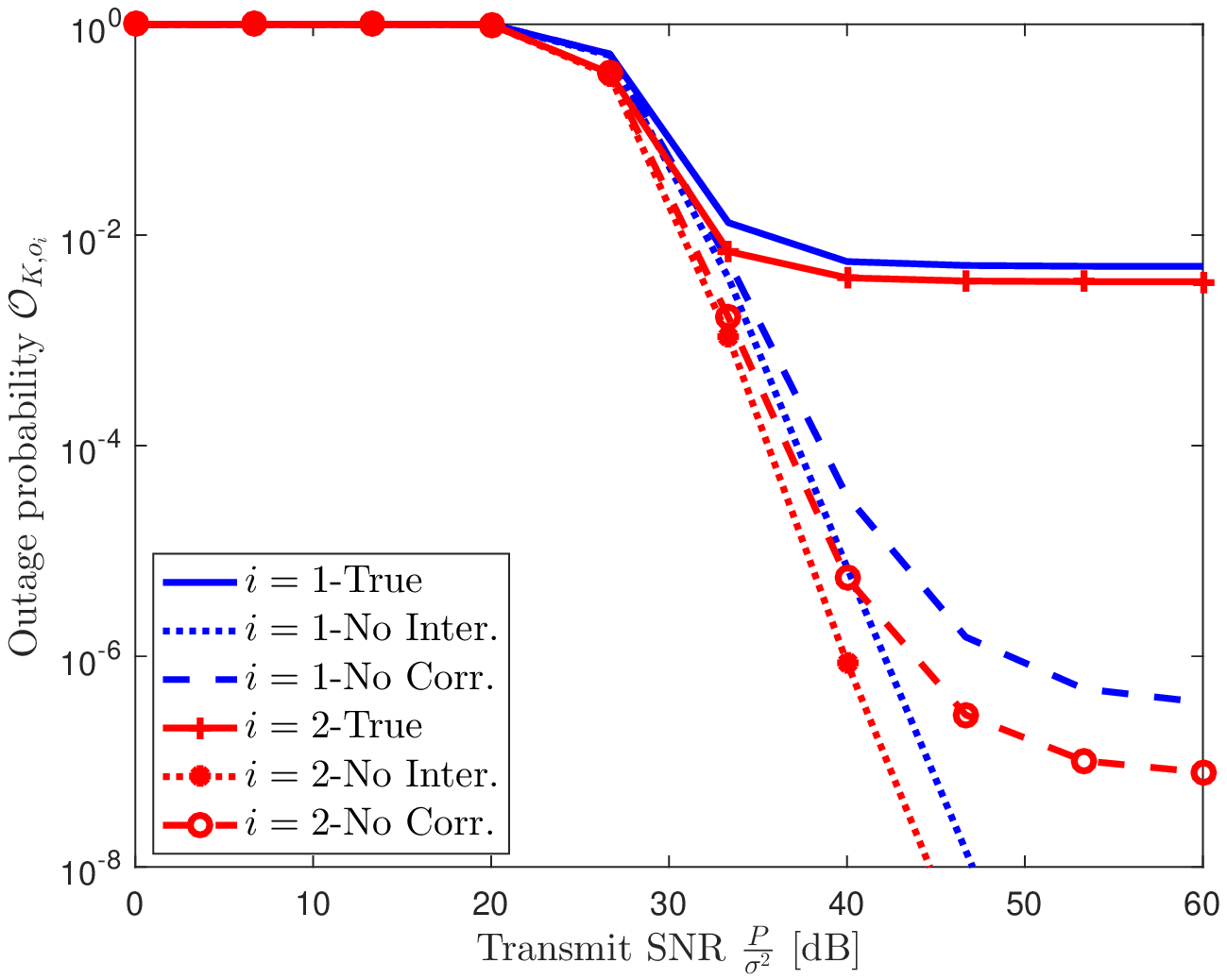}}\caption{\, Outage probability $\mathcal{O}_{K,o_i}$.}
\label{fig:sptepcorr_out}
    \end{subfigure}
       \caption{Effect of spatial and temporal interference correlation.}
\end{figure*}



\subsection{Comparison with non-cooperative HARQ assisted NOMA}
In order to quantify  the value of cooperation, the performance of the proposed scheme is compared with that of non-cooperative HARQ assisted NOMA scheme in this subsection by assuming $\frac{P}{\sigma^2} = 30$dB. 
It is worth noticing that the throughput and outage analyses in Section \ref{sec:per_ana} are also applicable to the non-cooperative HARQ assisted NOMA scheme by setting the transmit power at the relay (i.e., user 1) in phase II to zero. Figs. \ref{fig:comp_coopltat} and \ref{fig:comp_coopout} show the comparison between the two schemes in terms of the LTAT and the outage probability, respectively. It is readily seen in both figures that the proposed cooperative scheme outperforms the non-cooperative HARQ assisted NOMA scheme. For instance, the proposed scheme can reduce the outage probability by up to 32\% given $K=4$, compared with the non-cooperative HARQ assisted NOMA scheme. In addition, the LTAT and the outage probability $\mathcal{O}_{K,o_2}$ of the non-cooperative HARQ assisted NOMA scheme remain constant when the inter-device separation distance  $D$ varies, because the link between two NOMA users is not utilized to retransmit the message of user 2. Whereas the increase of $D$ will degrade the performance of the proposed scheme because of the rising path loss in relaying phase. It is worth noting that the outage probability of user 1 is independent of $D$ because user 2 does not decode nor relay user 1 message. Furthermore, Figs. \ref{fig:comp_coopltat} and \ref{fig:comp_coopout} also justify the accuracy of the approximate expressions in Theorem \ref{the:app}.


\begin{figure*}[t!]
    \centering
    \begin{subfigure}[t]{0.45\textwidth}
  \centerline{\includegraphics[width=  3 in]{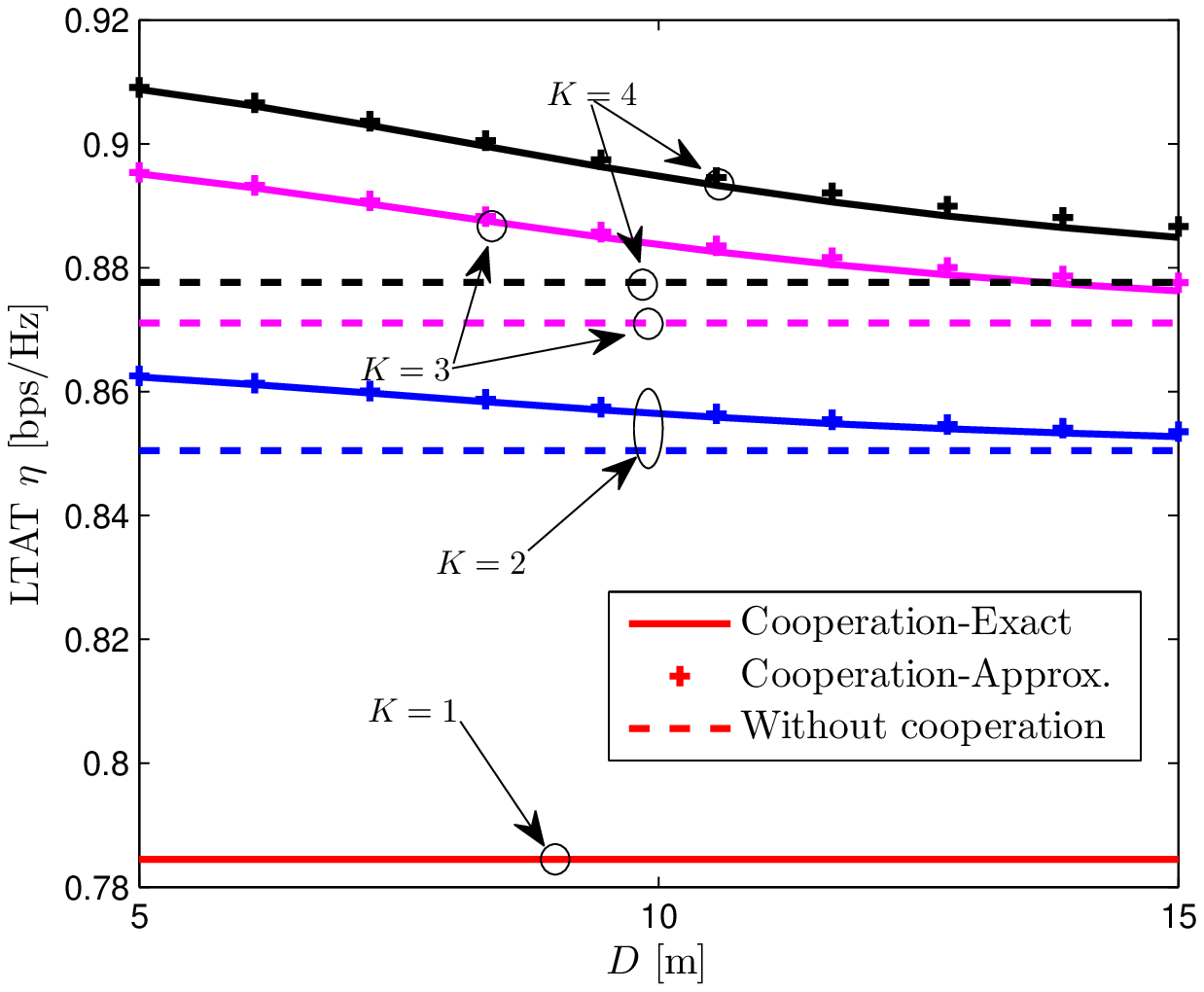}}
      \caption{\, LTAT $\eta$.}
\label{fig:comp_coopltat}
   \end{subfigure}
    ~
    \begin{subfigure}[t]{0.45\textwidth}
       \centerline{\includegraphics[width=  3.05 in]{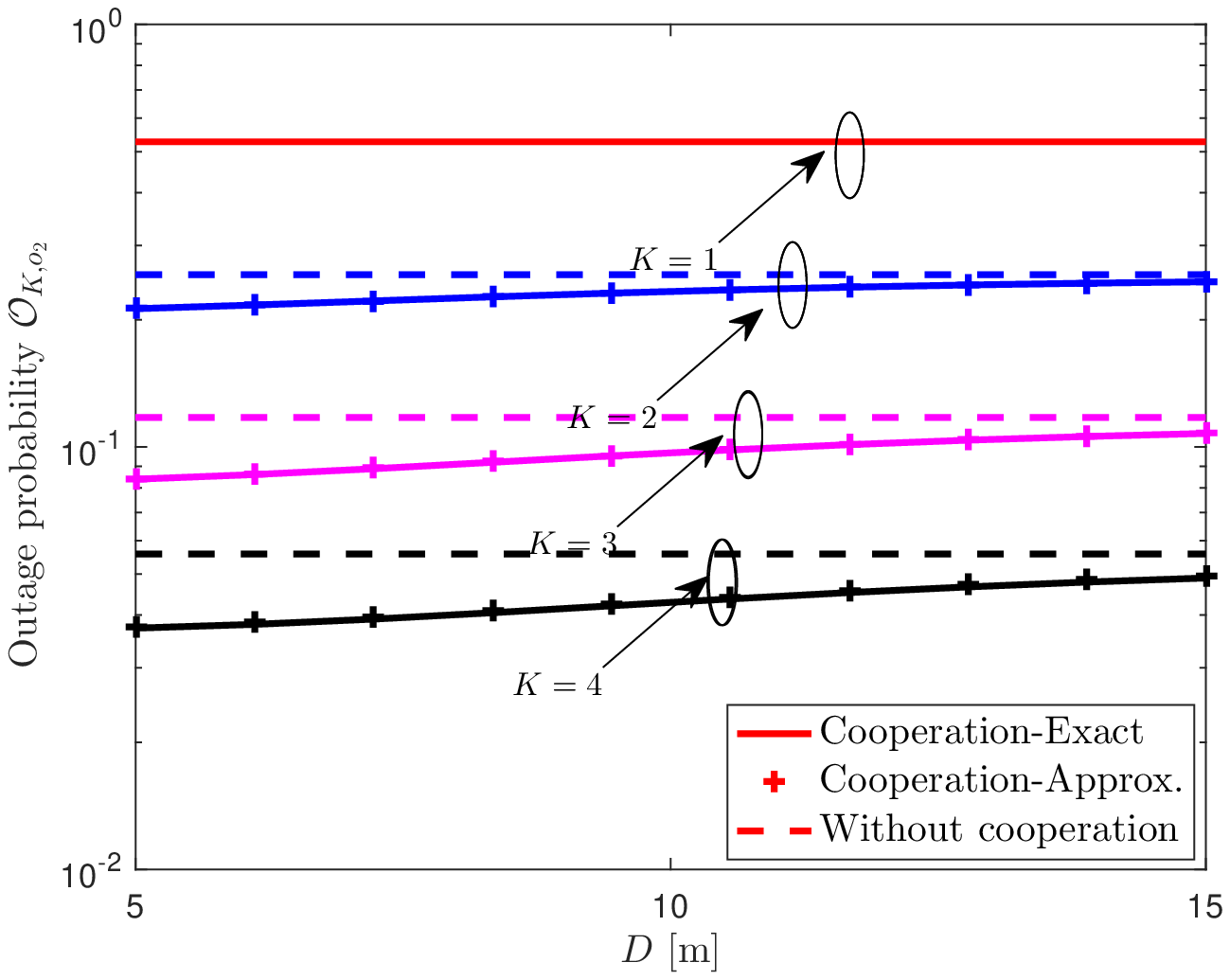}}\caption{\, Outage probability $\mathcal{O}_{K,o_i}$.}
\label{fig:comp_coopout}
    \end{subfigure}
       \caption{Effect of inter-user distance $D$.}
       \label{end_fig}
\end{figure*}



\subsection{Maximization of LTAT}
Figs. \ref{start_fig} - \ref{end_fig} confirm the high accuracy of the approximation approach of (\ref{eqn:varphi_0}). The simple and closed-form expression of (\ref{eqn:varphi_0}) enables the effective evaluation of outage probabilities compared with the double integral representation of (\ref{eqn:varphi_given}). Thus we apply the approximation approach to facilitate the optimal NOMA design in the sequel.

As seen in Fig. \ref{fig:ver_ltat}, the increase of the maximal number of HARQ transmissions may decrease the LTAT. In order to combat the negative impact of co-channel interference and fully exploit the benefit of cooperative HARQ, an interference aware optimal design is proposed herein. Particularly the LTAT is maximized through properly choosing system parameters while maintaining the quality of service. By taking the optimal rate selection as an example, the LTAT is maximized by optimally selecting transmission rates given the predetermined power allocation coefficient $\beta^2$, while guaranteeing outage constraints and the implementation of NOMA protocol. Mathematically, the optimization problem can be formulated as
\begin{equation}\label{eqn:opt_prob_simp}
\begin{array}{*{20}{cl}}
{\mathop {\rm maximize}\limits_{R_1, R_2} }&{\eta}\\
{{\rm{subject}}\,{\rm{to}}}&{\mathcal{O}_{K,o_i} \le \varepsilon_i},\, i= 1,2\\
{}&{0 \le \beta^2 < 2^{-R_2}},\\
\end{array}
\end{equation}
where $\varepsilon_i$ denotes the maximal allowable outage probability for user $i$.
For comparison, the HARQ assisted orthogonal multiple access (OMA) scheme is also implemented, where OMA scheme could be TDMA and Orthogonal frequency-division multiple access (OFDMA) \cite{saito2013non} etc. Unlike the proposed scheme, the HARQ assisted OMA transmission does not require the near user to decode the message of the far user first. Therefore, the near user can not help the source device deliver the message to exploit extra spatial diversity from cooperative communications. The LTAT of the HARQ assisted OMA scheme is derived in Appendix \ref{app:ltat_tdma}. For the fairness of the comparison, 
the same coefficient $\beta^2$ is introduced to allocate the orthogonal resources (bandwidth/time) in the OMA scheme. Moreover, we assume the same outage constraints for two users, i.e., $\varepsilon_1=\varepsilon_2=\varepsilon$.

Fig. \ref{fig:ver_rate_comp} manifests the superiority  of the optimal LTAT achieved by the proposed scheme over that of the OMA scheme under optimal rate selection. For instance, the proposed scheme yields an approximately $47$\% throughput gain when $\frac{P}{\sigma^2}=60$dB and $K=4$, compared with the OMA scheme. In addition, increasing the maximal number of transmissions is in favor of the optimal LTAT no matter under the proposed scheme or under the OMA scheme. It is important to note that the designs based on the `No Inter.' and `No Corr.' violate the outage probability constraints, and hence, the corresponding LTATs are not plotted in Fig. \ref{fig:ver_rate_comp}. Particularly, the outage probabilities $\mathcal O_{K,o_1}$ corresponding to $K=4$ for `No Inter.' and `No Corr.', respectively, are $0.2$ and $0.05$, which greatly exceed the outage constraint $\varepsilon=0.01$. Hence, totally ignoring the interference or just ignoring the interference correlation lead to an infeasible network design by violating the network operational constraints.  
To summarize, Fig. \ref{fig:ver_rate_comp} reveals the superior performance of the proposed interference aware design under the assumption of statistical CSI available at transmitter. It is worth noting that the same conclusion holds true if perfect CSI is known at transmitter \cite{ding2015cooperative}.
\begin{figure}
  \centering
  \includegraphics[width=3in,height=2.4in]{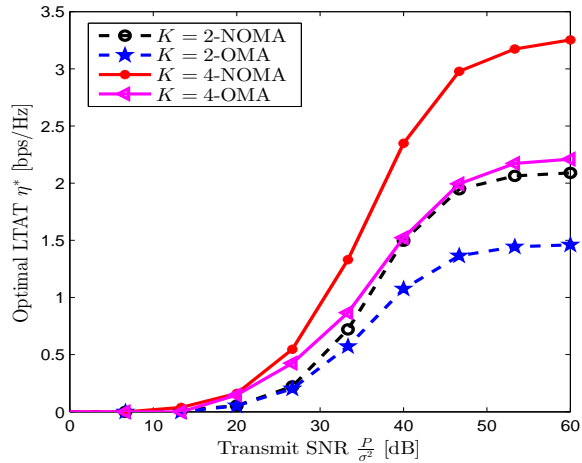}
  \caption{The maximal LTAT via the optimal rate selection with $\beta^2=0.3$ and $\varepsilon=0.01$.}\label{fig:ver_rate_comp}
\end{figure}

Furthermore, power allocation coefficient $\beta^2$ can also be optimally selected to maximize the LTAT given the desired transmission rates. However, it should be noticed that the joint optimal rate and power allocation for LTAT maximization may result in less or no information (or power) delivered (or allocated) to the far user with poor channel condition. Indeed, this is not beyond our expectation when user fairness (e.g., target transmission rate for each user) is not considered. Without fairness constraint, the joint optimization of the power and rate would aggressively allocate most of power to the user with better channel condition, which behaves like waterfilling algorithm regardless of user fairness \cite{ji2015power} and violates the intention of NOMA protocol \cite{timotheou2015fairness,cui2016novel,yang2016general,islam2016power}. This interesting phenomenon can be observed in Table \ref{tab:joint_power_rate} for the proposed scheme, where the notation ``$-$'' denotes no feasible solution.
Without any exceptions, the conclusion is also applicable to the OMA scheme. For further illustration, the joint power and rate optimization of the OMA system with $K=1$ is examined as an example in the following remark.
\begin{remark}\label{rem:joint_power}
For joint power and rate optimization of the OMA scheme to maximize the throughput with $K=1$ and $\varepsilon_1=\varepsilon_2=\varepsilon$, it is proved in Appendix \ref{app:joint_power} that no power would be allocated to convey information to the far user with worse channel condition, and the optimal transmission rate for the far user is zero, i.e., ${\beta^*}^2=1$ and ${R_2}^*=0$bps/Hz.
\end{remark}
\begin{table}[h!]
\centering
\caption{The optimal ${\beta}^2$ under joint power and rate optimization for $K=2$.}
\begin{tabular}{c||ccc}
  &\multicolumn{3}{c}{Transmit SNR $\frac{P}{\sigma^2}$}\\
  \hline
   Outage Tolerance $ \varepsilon$&0{dB} &30{dB}     & 60{dB}\\
   \hline
   $0.1$&-&1.0000 &0.9999\\
   $0.01$&-&1.0000 &1.0000\\
\end{tabular}
\label{tab:joint_power_rate}
\end{table}
\subsection{Maximization of ASE}
Aside from the LTAT, the ASE is another useful metric to characterize the performance of the whole D2D network \cite{andrews2010primer}. 
Specifically, the ASE of the D2D network is given by
\begin{equation}\label{eqn:network_througput}
\Delta  = \lambda \eta .
\end{equation}
Inspired by (\ref{eqn:opt_prob_simp}), the intensity of D2D transmitters can also be jointly designed to maximize the ASE, such that
\begin{equation}\label{eqn:network_throghput_simp}
\begin{array}{*{20}{cl}}
{\mathop {\rm maximize}\limits_{R_1, R_2,\lambda} }&{\Delta}\\
{{\rm{subject}}\,{\rm{to}}}&{\mathcal{O}_{K,o_i} \le \varepsilon_i},\, i= 1,2\\
{}&{0 \le \beta^2 < 2^{-R_2}}.\\
\end{array}
\end{equation}

In Fig. \ref{fig:net_rate_comp}, the optimal ASE is plotted against the transmit SNR via optimal design of transmission rates and intensity. It is observed in Fig. \ref{fig:net_rate_comp} that increasing $K$ and relaxing $\varepsilon$ could significantly improve the optimal ASE. Moreover, it is intuitive that increasing the transmit SNR will enhance the optimal ASE. However, the gain turns out to be negligible in high SNR regime. This is because  increasing the transmit SNR not only improves the received SNR but also boosts the interference, and consequently SINR does not vary. Similar to the LTAT scenario, ASE maximization based on `No Corr.' model would violate the outage probability constraints, and hence, lead to an infeasible solution.

\begin{figure}
  \centering
  \includegraphics[width=3in,height=2.4in]{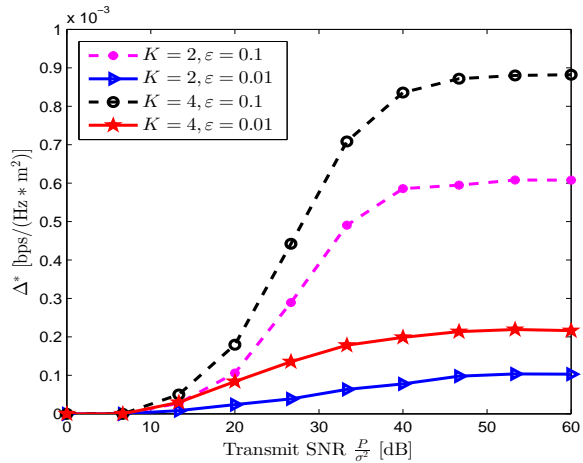}
  \caption{The maximal ASE via optimizing the transmission rates and the intensity.}\label{fig:net_rate_comp}
\end{figure}

\section{Conclusion}\label{sec:cond}

This paper has utilized stochastic geometry to develop an interference-aware mathematical model for cooperative HARQ assisted NOMA in D2D networks. Particularly, by modeling the spatial locations of the interfering devices using a PPP, tractable exact expressions for the long term average throughput (LTAT) and outage probability of a two user NOMA scenario have been derived. The developed model accounts for the spatial and temporal correlation of the interference at the two NOMA users and across the HARQ rounds. It has been shown that the decoding performance at the two receivers are interwoven and that the temporal and spatial correlation negatively influences the NOMA performance. To this end, an accurate analytical approximation for the LTAT has been proposed to enable interference aware optimal network design. Numerical results have shown that the proposed cooperative NOMA scheme decreases the outage probability by up to $32$\% compared to the non-cooperative case. Additionally, the optimized LTAT of proposed scheme outperforms that of the OMA scheme by $47$\%. It has also been shown that interference-oblivious rate selection results in violating the network outage constraints. Finally, optimal ASE has been presented to maximize the overall network performance.

\appendices
\section{Proof of (\ref{eqn:out_in_ex_manifin})}\label{app:outage_o1_1_f}
Putting (\ref{eqn:sinr_1_2}), (\ref{eqn:sinr_1_1}) and (\ref{eqn:sinr_phase2_o1}) into (\ref{eqn:outage_o1_1}), and then rearranging it yields
\begin{align}\label{eqn:out_o11p_simp}
&\mathbb P\left[{{{\bar \Theta }_{{o_1},1}},{\Theta _{{o_1},2,l}},{\Theta _{{o_2},k}}} \right]
= \mathbb P \left[ {\begin{array}{*{20}{c}}
\begin{array}{l}
\bigcap\limits_{j = 1}^{l - 1} \frac{{P{{\left| {{h_{z{o_1},j}}} \right|}^2}}}{{{I_{1,j}} + {\sigma ^2}}} < {U_{{o_1},2}},
{U_{{o_1},2}} \le \frac{{P{{\left| {{h_{z{o_1},l}}} \right|}^2}}}{{{I_{1,l}} + {\sigma ^2}}} < {U_{{o_1},1,I}} ,\\
\bigcap\limits_{j = l + 1}^k {\frac{{P{{\left| {{h_{z{o_1},j}}} \right|}^2}}}{{{I_{1,j}} + {\sigma ^2}}} < {U_{{o_1},1,I}}} ,
\bigcap\limits_{j = k + 1}^K {\frac{{P{{\left| {{h_{z{o_1},j}}} \right|}^2}}}{{{I_{1,j}} + {\sigma ^2}}} < {U_{{o_1},II}}},
\end{array}\\
{ \bigcap\limits_{j = 1}^{k - 1} {\frac{{P{{\left| {{h_{z{o_2},j}}} \right|}^2}}}{{{I_{2,j}} + {\sigma ^2}}} < {U_{{o_2},I}}} ,\frac{{P{{\left| {{h_{z{o_2},k}}} \right|}^2}}}{{{I_{2,k}} + {\sigma ^2}}} \ge {U_{{o_2},I}}}
\end{array}} \right],
\end{align}
where ${U_{{o_1},1,I}} = \frac{{{2^{{R_1}}} - 1}}{{{\beta ^2}\ell \left( {{d_1}} \right)}}$, ${U_{{o_1},2}} = \frac{{{2^{{R_2}}} - 1}}{{\left( {1 - {2^{{R_2}}}{\beta ^2}} \right)\ell \left( {{d_1}} \right)}}$, ${U_{{o_2},I}} = \frac{{{2^{{R_2}}} - 1}}{{\left( {1 - {2^{{R_2}}}{\beta ^2}} \right)\ell \left( {{d_2}} \right)}}$ and ${U_{{o_1},II}} = \frac{{{2^{{R_1}}} - 1}}{{\ell \left( {{d_1}} \right)}}$. Herein, it should be mentioned that ${1 - {2^{{R_2}}}{\beta ^2}} > 0$, otherwise user 1 is unable to mitigate the NOMA interference ${\bf s}_2$. For simplicity, we define the following successful decoding events ${A_{1,j}} \triangleq \left\{ {\frac{{P{{\left| {{h_{z{o_1},j}}} \right|}^2}}}{{{I_{1,j}} + {\sigma ^2}}} \ge  {U_{{o_1},2}}} \right\}$, $ j \in [1,l]$;
${A_{2,j}} \triangleq \left\{ {\frac{{P{{\left| {{h_{z{o_1},j}}} \right|}^2}}}{{{I_{1,j}} + {\sigma ^2}}} \ge {U_{{o_1},1,I}}} \right\}$, $j \in [l,k]$;
${A_{3,j}} \triangleq \left\{ {\frac{{P{{\left| {{h_{z{o_1},j}}} \right|}^2}}}{{{I_{1,j}} + {\sigma ^2}}} \ge {U_{{o_1},II}}} \right\}$, $j \in [k+1,K]$;
${B_{1,j}} \triangleq \left\{ {\frac{{P{{\left| {{h_{z{o_2},j}}} \right|}^2}}}{{{I_{2,j}} + {\sigma ^2}}} \ge {U_{{o_2},I}}} \right\}$, $j \in [1,k]$.
%
%
%
Then (\ref{eqn:out_o11p_simp}) can be simplified as
\begin{align}\label{eqn:outo11p_rew}
&\mathbb P\left[ {{{\bar \Theta }_{{o_1},1}},{\Theta _{{o_1},2,l}} ,{\Theta _{{o_2},k}}} \right]  = \mathbb P\left[ \begin{array}{l}
\bigcap\limits_{j = 1}^{l - 1} {\left( {\Omega - {A_{1,j}}} \right)} ,{A_{1,l}} - {A_{2,l}},\bigcap\limits_{j = l + 1}^k {\left( {\Omega - {A_{2,j}}} \right)} ,\\
\bigcap\limits_{j = k + 1}^K {\left( {\Omega - {A_{3,j}}} \right)} ,\bigcap\limits_{j = 1}^{k - 1} {\left( {\Omega - {B_{1,j}}} \right)} ,{B_{1,k}}
\end{array} \right].
\end{align}
By applying inclusion-exclusion principle into the first term in the square bracket, (\ref{eqn:outo11p_rew}) can be rewritten as
\begin{equation}\label{eqn:out11_inc_exp00}
\mathbb P\left[ {{{\bar \Theta }_{{o_1},1}},{\Theta _{{o_1},2,l}} ,{\Theta _{{o_2},k}}} \right]
=\sum\limits_{{\tau _1} = 0}^{l - 1} {\sum\limits_{{\nu _1} \in {{\cal M}_{{\tau _1}}}} {{{\left( { - 1} \right)}^{{\tau _1}}}} }
\mathbb P\left[ {\begin{array}{*{20}{l}}
{\bigcap\limits_{v \in {\nu _1}} {{A_{1,v}}} ,{A_{1,l}} - {A_{2,l}},\bigcap\limits_{j = l + 1}^k {\left( {\Omega - {A_{2,j}}} \right)} ,}\\
{\bigcap\limits_{j = k + 1}^K {\left( {\Omega - {A_{3,j}}} \right)} ,\bigcap\limits_{j = 1}^{k - 1} {\left( {\Omega - {B_{1,j}}} \right)} ,{B_{1,k}}}
\end{array}} \right ],
\end{equation}
where $\mathcal M_{\tau_1}$ refers to the union of all the $\tau_1$-element subsets of the natural number set $\{1,2,\cdots,l-1\}$. Similarly, repeatedly using the same approach as (\ref{eqn:out11_inc_exp00}) leads to 
\begin{multline}\label{eqn:out11_inc_exp0}
\mathbb P\left[ {{{\bar \Theta }_{{o_1},1}},{\Theta _{{o_1},2,l}} ,{\Theta _{{o_2},k}}} \right]
= \sum\limits_{{\tau _1} = 0}^{l - 1} {\sum\limits_{{\tau _2} = 0}^{k - l} {\sum\limits_{{\tau _3} = 0}^{K - k} {\sum\limits_{{\tau _4} = 0}^{k - 1} {\sum\limits_{{\nu _1} \in {{\cal M}_{{\tau _1}}}} {\sum\limits_{{\nu _2} \in {{\cal M}_{{\tau _2}}}} {\sum\limits_{{\nu _3} \in {{\cal M}_{{\tau _3}}}} {\sum\limits_{{\nu _4} \in {{\cal M}_{{\tau _4}}}} {{{\left( { - 1} \right)}^{\sum\limits_{j = 1}^4 {{\tau _j}} }}} \times} } } } } } } \\
  \mathbb P\left[ {\bigcap\limits_{v \in {\nu _1}} {{A_{1,v}}} ,{A_{1,l}} - {A_{2,l}},\bigcap\limits_{v \in {\nu _2}} {{A_{2,v}}} ,\bigcap\limits_{v \in {\nu _3}} {{A_{3,v}}} ,\bigcap\limits_{v \in {\nu _4}} {{B_{1,v}}} ,{B_{1,k}}} \right],
\end{multline}
where $\mathcal M_{\tau_2}$, $\mathcal M_{\tau_3}$, $\mathcal M_{\tau_4}$ denote unions of all the subsets of natural number sets $\{l+1,\cdots,k\}$, $\{k+1,\cdots,K\}$ and $\{1,\cdots,k-1\}$ with cardinalities ${\tau_2}$, ${\tau_3}$ and ${\tau_4}$, respectively.

Noticing that ${{A_{2,l}}} \subset {A_{1,l}}$ if $U_{o_1,1,I} > U_{o_1,2}$, otherwise ${{A_{1,l}}} - {A_{2,l}} = \emptyset$. Thereafter, (\ref{eqn:out11_inc_exp0}) can be derived as
\begin{multline}\label{eqn:out_in_ex_mani0}
\mathbb P\left[ {{{\bar \Theta }_{{o_1},1}},{\Theta _{{o_1},2,l}} ,{\Theta _{{o_2},k}}} \right]  = \sum\limits_{{\tau _1} = 0}^{l - 1} {\sum\limits_{{\tau _2} = 0}^{k - l} {\sum\limits_{{\tau _3} = 0}^{K - k} {\sum\limits_{{\tau _4} = 0}^{k - 1} {\sum\limits_{{\nu _1} \in {\mathcal M_{{\tau _1}}}} {\sum\limits_{{\nu _2} \in {\mathcal M_{{\tau _2}}}} {\sum\limits_{{\nu _3} \in {\mathcal M_{{\tau _3}}}} {\sum\limits_{{\nu _4} \in {\mathcal M_{{\tau _4}}}} {{{\left( { - 1} \right)}^{\sum\limits_{j = 1}^4 {{\tau _j}} }}} \times} } } } } } } \\
 {\left[ \begin{array}{l}
\mathbb P\left[{\bigcap\limits_{v \in {\nu _1}} {{A_{1,v}}} ,{A_{1,l}},\bigcap\limits_{v \in {\nu _2}} {{A_{2,v}}} ,\bigcap\limits_{v \in {\nu _3}} {{A_{3,v}}} ,\bigcap\limits_{v \in {\nu _4}} {{B_{1,v}}} ,{B_{1,k}}} \right]\\
 - \mathbb P\left[ {\bigcap\limits_{v \in {\nu _1}} {{A_{1,v}}} ,{A_{2,l}},\bigcap\limits_{v \in {\nu _2}} {{A_{2,v}}} ,\bigcap\limits_{v \in {\nu _3}} {{A_{3,v}}} ,\bigcap\limits_{v \in {\nu _4}} {{B_{1,v}}} ,{B_{1,k}}} \right]
\end{array} \right]^ + }.
\end{multline}

It follows from (\ref{eqn:out_in_ex_mani0}) that all the terms in square brackets have the same sign no matter $U_{o_1,1,I} > U_{o_1,2}$ or not. 
Hence, (\ref{eqn:out_in_ex_mani0}) can be rewritten as
\begin{equation}\label{eqn:out_in_ex_mani}
\mathbb P\left[ {{{\bar \Theta }_{{o_1},1}},{\Theta _{{o_1},2,l}} ,{\Theta _{{o_2},k}}} \right]   =
{\left[ \begin{array}{l}
\sum\limits_{{\tau _1} = 0}^{l - 1} {\sum\limits_{{\tau _2} = 0}^{k - l} {\sum\limits_{{\tau _3} = 0}^{K - k} {\sum\limits_{{\tau _4} = 0}^{k - 1} {\sum\limits_{{\nu _1} \in {{\cal M}_{{\tau _1}}}} {\sum\limits_{{\nu _2} \in {{\cal M}_{{\tau _2}}}} {\sum\limits_{{\nu _3} \in {{\cal M}_{{\tau _3}}}} {\sum\limits_{{\nu _4} \in {{\cal M}_{{\tau _4}}}} {{{\left( { - 1} \right)}^{\sum\limits_{j = 1}^4 {{\tau _j}} }}}  \times } } } } } } } \\
\left( {\begin{array}{*{20}{l}}
\mathbb P{\left[ {\bigcap\limits_{v \in {\nu _1}} {{A_{1,v}}} ,{A_{1,l}},\bigcap\limits_{v \in {\nu _2}} {{A_{2,v}}} ,\bigcap\limits_{v \in {\nu _3}} {{A_{3,v}}} ,\bigcap\limits_{v \in {\nu _4}} {{B_{1,v}}} ,{B_{1,k}}} \right]}\\
{ - \mathbb P\left[ {\bigcap\limits_{v \in {\nu _1}} {{A_{1,v}}} ,{A_{2,l}},\bigcap\limits_{v \in {\nu _2}} {{A_{2,v}}} ,\bigcap\limits_{v \in {\nu _3}} {{A_{3,v}}} ,\bigcap\limits_{v \in {\nu _4}} {{B_{1,v}}} ,{B_{1,k}}} \right]}
\end{array}} \right)
\end{array} \right]^ + },
\end{equation}

Noticing the inner probability terms $\mathbb P(\cdot)$ are independent of $\nu_1$, $\nu_2$, $\nu_3$ and $\nu_4$, the cardinalities of set $\mathcal M_{\tau_1}, \mathcal M_{\tau_2}, \mathcal M_{\tau_3}$ and $\mathcal M_{\tau_4}$ are given by $C_{l - 1}^{{\tau _1}}$, $C_{k-l}^{{\tau _2}}$, $C_{K - k}^{{\tau _3}}$ and $C_{k - 1}^{{\tau _4}}$, respectively. Accordingly, (\ref{eqn:out_in_ex_mani}) can be simplified as (\ref{eqn:out_in_ex_manifin}), wherein $\Psi ({\bf{U}},\bs \tau ;\hat {\bf{U}},\hat {\bs \tau} )$ is defined for notational convenience as
\begin{equation}\label{eqn:genr_decoding_scu}
\Psi ({\bf{U}},\bs \tau ;\hat {\bf{U}},\hat {\bs \tau} ) \triangleq \mathbb P\left[ {\bigcap\limits_{n = 1}^N {\bigcap\limits_{k = 1}^{\tau_n } {{\mathcal A_{n,k}}} } ,\bigcap\limits_{n = 1}^M {\bigcap\limits_{k =  1}^{\hat \tau _n } {{\mathcal B_{n,k}}} } } \right],
\end{equation}
where ${\bf U} = (U_1,\cdots,U_N)$, ${\bs \tau} = (\tau_1,\cdots,\tau_N)$, $\hat{\bf U} = (\hat U_1,\cdots,\hat U_M)$, $\hat{\bs \tau} = (\hat \tau_1,\cdots,\hat \tau_M)$, ${\mathcal A_{n,k}} \triangleq \left\{ \frac{{P{{\left| {{h_{k_n}}} \right|}^2}}}{{{I'_{1,{k_n}}} + {\sigma ^2}}} \ge {U_n}\right\} $, ${\mathcal B_{n,k}} \triangleq \left\{ \frac{{P{{\left| {{{\hat h}_{k_n}}} \right|}^2}}}{{{I'_{2,{k_n}}} + {\sigma ^2}}} \ge {{\hat U}_n}\right\} $, ${I'_{i,k_n}} = P\sum\nolimits_{x \in \Phi \backslash \left\{ z \right\}} {\ell \left( {\left\| {x - {o_i}} \right\|} \right){{\left| {{h'_{xo_i,k_n}}} \right|}^2}}$ and ${k_n} = \sum\nolimits_{\iota  = 1}^{n - 1} {{\tau _\iota }}  + k$, the channel amplitudes $\left\{\left|h_k\right|,\, k=[1,\cdots,{\sum\limits_{\iota  = 1}^N \tau_\iota}]\right\}$, $\left\{\left|\hat h_k\right|,\, k=[1,\cdots,{\sum\limits_{\iota  = 1}^M \hat \tau_\iota}]\right\}$,  $\left\{\left|{h'_{xo_1,k}}\right|,\, k=[1,\cdots,{\sum\limits_{\iota  = 1}^N \tau_\iota}]\right\}$ and $\left\{\left|{h'_{xo_2,k}}\right|,\, k=[1,\cdots,{\sum\limits_{\iota  = 1}^M \hat \tau_\iota}]\right\}$
follow independent Rayleigh distributions with unit average power.


Since fading channels follow independent Rayleigh distribution given the interferences $I'_{i,k}$, $\Psi ({\bf{U}},\bs \tau ;\hat {\bf{U}},\hat {\bs \tau} )$ can be derived as 
\begin{align}\label{eqn:Psi_de}
 \Psi ({\bf{U}},\bs \tau ;\hat {\bf{U}},\hat {\bs \tau} )& ={\mathbb E_{{I'_{i,k}}}} \left(\mathbb P\left[ {\bigcap\limits_{n = 1}^N {\bigcap\limits_{k = 1}^{{\tau _n }} {{{\left| {{h_{k_n}}} \right|}^2} \ge \frac{{{U_n}\left( {{I'_{1,k_n}} + {\sigma ^2}} \right)}}{P}} } ,\bigcap\limits_{n = 1}^M {\bigcap\limits_{k = 1}^{{{\hat \tau _n }} } {{{\left| {{{\hat h}_{k_n}}} \right|}^2} \ge \frac{{{{\hat U}_n}\left( {{I'_{2,k_n}} + {\sigma ^2}} \right)}}{P}} } } \right]\right)\notag\\
   & = {e^{ - \frac{{{\sigma ^2}}}{P}
   \left({\bf U}^{\rm T}{\bs \tau} + {{\hat {\bf U}}^{\rm T}{\hat {\bs \tau} }}\right)
   }}{\mathbb E_{{I'_{i,k_n}}}}\left( {\prod\limits_{n = 1}^N {\prod\limits_{k = 1}^{{{\tau _n }} } {{e^{ - \frac{{{U_n}{I'_{1,k_n}}}}{P}}}} } \prod\limits_{n = 1}^M {\prod\limits_{k = 1}^{ {{\hat \tau _n }} } {{e^{ - \frac{{{{\hat U}_n}{I'_{2,k_n}}}}{P}}}} } } \right).
\end{align}
Then plugging (\ref{eqn:inter_fer_1}) into (\ref{eqn:Psi_de}) leads to
\begin{align}\label{eqn:psi_de_sub}
&\Psi ({\bf{U}},\bs \tau ;\hat {\bf{U}},\hat {\bs \tau} ) = {e^{ - \frac{{{\sigma ^2}}}{P}\left( {{{\bf{U}}^{\rm{T}}}\bs \tau  + {{\hat {\bf{U}}}^{\rm{T}}}\hat {\bs\tau} } \right)}}{\mathbb E_{{I'_{i,k}}}}\left( \begin{array}{l}
\prod\limits_{n = 1}^N {\prod\limits_{k = 1}^{{\tau _n}} {{e^{ - {U_n}{\sum _{x \in \Phi \backslash \left\{ z \right\}}}\ell \left( {\left\| {x - {o_1}} \right\|} \right){{\left| {{h'_{x{o_1},{k_n}}}} \right|}^2}}}} }  \\
\times\prod\limits_{n = 1}^M {\prod\limits_{k = 1}^{{{\hat \tau }_n}} {{e^{ - {{\hat U}_n}{\sum _{x \in \Phi \backslash \left\{ z \right\}}}\ell \left( {\left\| {x - {o_2}} \right\|} \right){{\left| {{h'_{x{o_2},{k_n}}}} \right|}^2}}}} }
\end{array}\right).
\end{align}
Given the PPP $\Phi$ and noticing the independence of fading channels, (\ref{eqn:psi_de_sub}) can be further written as 
\begin{equation}\label{eqn:psi_de_sub_split}
\Psi ({\bf{U}},\bs \tau ;\hat {\bf{U}},\hat {\bs \tau} ) = {e^{ - \frac{{{\sigma ^2}}}{P}\left( {{{\bf{U}}^{\rm{T}}}\bs \tau  + {{\hat {\bf{U}}}^{\rm{T}}}\hat {\bs\tau} } \right)}} \mathbb E_z^!\left( \prod\limits_{x \in \Phi \backslash \left\{ z \right\}} {{ {{\left( \begin{array}{l}
\prod\limits_{n = 1}^N {\prod\limits_{k = 1}^{{{\tau _n }} } {{\mathbb E_{\left| {{h'_{x{o_1},k_n}}} \right|}}\left( {{e^{ - {U_n}\ell \left( {\left\| {x - {o_1}} \right\|} \right){{\left| {{h'_{x{o_1},k_n}}} \right|}^2}}}} \right)} } \\
\prod\limits_{n = 1}^M {\prod\limits_{k = 1}^{ {{\hat \tau _n }} } {{\mathbb E_{\left| {{h'_{x{o_2},k_n}}} \right|}}\left( {{e^{ - {{\hat U}_n}\ell \left( {\left\| {x - {o_2}} \right\|} \right){{\left| {{h'_{x{o_2},k_n}}} \right|}^2}}}} \right)} }
\end{array} \right)}} }}  \right).
\end{equation}
where $\mathbb E_z^!$ denotes the expectation taken against the reduced Palm distribution of the PPP $\Phi$. Averaging over channel coefficients of Rayleigh distribution yields
\begin{equation}\label{eqn:psi_de_sub_splitavg}
\Psi ({\bf{U}},\bs \tau ;\hat {\bf{U}},\hat {\bs \tau} ) = {e^{ - \frac{{{\sigma ^2}}}{P}\left( {{{\bf{U}}^{\rm{T}}}\bs \tau  + {{\hat {\bf{U}}}^{\rm{T}}}\hat {\bs\tau} } \right)}}
 \mathbb E_z^!\left( {\begin{array}{l}
{\prod \limits_{x \in \Phi \backslash \left\{ z \right\}}}\prod\limits_{n = 1}^N {\frac{1}{{{{\left( {1 + {U_n}\ell \left( {\left\| {x - {o_1}} \right\|} \right)} \right)}^{{\tau _n}}}}}}
\prod\limits_{n = 1}^M {\frac{1}{{{{\left( {1 + {{\hat U}_n}\ell \left( {\left\| {x - {o_2}} \right\|} \right)} \right)}^{{{\hat \tau }_n}}}}}}
\end{array}} \right).
\end{equation}

It follows by using Slivnyak theorem and the Laplace functional of PPPs that \cite{haenggi2012stochastic}
\begin{equation}\label{eqn:psi_de_sub_splitavgsl}
\Psi ({\bf{U}},\bs \tau ;\hat {\bf{U}},\hat {\bs \tau} ) = {e^{ - \frac{{{\sigma ^2}}}{P}\left( {{{\bf{U}}^{\rm{T}}}\bs \tau  + {{\hat {\bf{U}}}^{\rm{T}}}\hat {\bs\tau} } \right)}}
 {{e^{ - \int\limits_{{\mathbb R^2}} {\left( {1 - \prod\limits_{n = 1}^N {\frac{1}{{{{\left( {1 + {U_n}\ell \left( {\left\| {x - {o_1}} \right\|} \right)} \right)}^{{\tau _n}}}}}} \prod\limits_{n = 1}^M {\frac{1}{{{{\left( {1 + {{\hat U}_n}\ell \left( {\left\| {x - {o_2}} \right\|} \right)} \right)}^{{{\hat \tau }_n}}}}}} } \right)\mathbb P_z^!\left[ {dx} \right]} }}},
\end{equation}
where $\mathbb P_z^!$ denotes the reduced Palm distribution of the PPP $\Phi$ with intensity $\lambda$. By making the change of variables and after some algebraic manipulations, (\ref{eqn:psi_de_sub_splitavgsl}) can be derived as (\ref{eqn:genr_decoding_scu1}).
\section{Proof of (\ref{eqn:varphi_0})}\label{app:proof_o1eqo2}
For small $D$, $\varphi ({\bf{U}},{\bs \tau} ;{\bf{\hat U}},\hat {\bs \tau} )$ in (\ref{eqn:varphi_given}) can be approximated by setting $o_1=o_2$ as 
\begin{align}\label{eqn:varphi_11}
\varphi ({\bf{U}},{\bs \tau} ;{\bf{\hat U}},\hat {\bs \tau} ) \approx \varphi (\tilde{\bf{U}},\tilde{\bs \tau};{\bf{0}},{\bf{0}} ),
\end{align}
where $\tilde{\bf U} = ({\bf U},\hat{\bf U})=(\tilde U_1,\cdots,\tilde U_{N+M})$ and $\tilde{\bs \tau} = ({\bs \tau},\hat{\bs \tau})=(\tilde \tau_1,\cdots,\tilde \tau_{N+M})$.
By applying polar coordinates together with (\ref{eqn:varphi_given}), $\varphi (\tilde{\bf{U}},\tilde{\bs \tau};{\bf{0}},{\bf{0}} )$ can be rewritten as
\begin{equation}\label{eqn:varphi_1}
\varphi (\tilde{\bf{U}},\tilde{\bs \tau};{\bf{0}},{\bf{0}} ) = \int\limits_0^\infty  {\int\limits_0^{2\pi } {\left( {1 - \prod\limits_{\iota  = 1}^{N+M} {{{\left( {1 + \frac{{{\tilde U_\iota }}}{{{r^\alpha }}}} \right)}^{ - {\tilde \tau _\iota }}}} } \right)rd\theta dr} }
 = 2\pi \int\limits_0^\infty  {\left( {1 - \prod\limits_{\iota  = 1}^{N+M} {{{\left( {1 + \frac{{{\tilde U_\iota }}}{{{r^\alpha }}}} \right)}^{ - {\tilde \tau _\iota }}}} } \right)rdr}.
\end{equation}
By using integration by parts, (\ref{eqn:varphi_1}) can then be derived as
\begin{multline}\label{eqn:integral_p_intbyparts}
\varphi (\tilde{\bf{U}},\tilde{\bs \tau};{\bf{0}},{\bf{0}} )  = \pi \int\limits_0^\infty  {\left( {1 - \prod\limits_{\iota  = 1}^{N+M} {{{\left( {1 + \frac{{{\tilde U_\iota }}}{{{r^\alpha }}}} \right)}^{ - {\tilde \tau_\iota }}}} } \right)d{r^2}} \\
=\pi\sum\limits_{\kappa  = 1}^{N+M} {{{{\tilde \tau _\kappa }\alpha {\tilde U_\kappa }}} }
 \int\limits_0^\infty  {{r^{1 - \alpha }}{{\left( {1 + \frac{{{\tilde U_\kappa }}}{{{r^\alpha }}}} \right)}^{ - {\tilde \tau _\kappa } - 1}}\prod\limits_{\scriptstyle\iota  = 1\hfill\atop
\scriptstyle\iota  \ne \kappa \hfill}^{N+M} {{{\left( {1 + \frac{{{\tilde U_\iota }}}{{{r^\alpha }}}} \right)}^{ - {\tilde \tau _\iota }}}} dr}\\
 =\pi \sum\limits_{\kappa  = 1}^{N+M} {{{{\tilde \tau _\kappa }\alpha {\tilde U_\kappa }}}\int\limits_0^\infty  {{r^{1 - \alpha }}\prod\limits_{\iota  = 1}^{N+M} {{{\left( {1 + \frac{{{\tilde U_\iota }}}{{{r^\alpha }}}} \right)}^{ - {\tilde \tau _\iota } - {\delta _{\iota  - \kappa }}}}} dr} },
\end{multline}
where ${{\delta }_s}$ denotes Dirac function. By making a change of variable $s=r^\alpha$, we have
\begin{equation}\label{eqn:integral_change_r2s}
\varphi (\tilde{\bf{U}},\tilde{\bs \tau};{\bf{0}},{\bf{0}} ) =\pi\sum\limits_{\kappa  = 1}^{N+M} {{{{\tilde \tau _\kappa }{\tilde U_\kappa }}}}
\int\limits_0^\infty  {{s^{\sum\limits_{\iota  = 1}^{N+M} {{\tilde \tau _\iota }}  + \frac{2}{\alpha } - 1}}\prod\limits_{\iota  = 1}^{N+M} {{{\left( {s + {\tilde U_\iota }} \right)}^{ - {\tilde \tau _\iota } - {\delta _{\iota  - \kappa }}}}} ds}.
\end{equation}

Assuming that ${\tilde U_\mu } = \max \left\{ {{\tilde U_1}, \cdots ,{\tilde U_{N+M}}} \right\}$ and by the change of variable $z = \frac{{{\tilde U_\mu }}}{{s + {\tilde U_\mu }}}$, (\ref{eqn:integral_change_r2s}) can be rewritten as
\begin{equation}\label{eqn:integral_changeofvari_hy}
\varphi (\tilde{\bf{U}},\tilde{\bs \tau};{\bf{0}},{\bf{0}} )
 = \pi\sum\limits_{\kappa  = 1}^{N+M} {{{{\tilde \tau _\kappa }{\tilde U_\kappa }{\tilde U_\mu }^{\frac{2}{\alpha } - 1}}} }\int\limits_0^1 {{z^{ - \frac{2}{\alpha }}}{{\left( {1 - z} \right)}^{\sum\limits_{\iota  = 1}^{N+M} {{\tilde \tau _\iota }}  + \frac{2}{\alpha } - 1}}}\prod\limits_{\scriptstyle\iota  = 1\hfill\atop
\scriptstyle\iota  \ne \mu \hfill}^{N+M} {{{\left( {1 - \left( {1 - \frac{{{\tilde U_\iota }}}{{{\tilde U_\mu }}}} \right)z} \right)}^{ - {\tilde \tau _\iota } - {\delta _{\iota  - \kappa }}}}} dz.
\end{equation}
With the definition of the fourth kind of Lauricella function $F_D^{(N)}(\cdot)$ in \cite[Eq. A.52]{mathai2009h}, (\ref{eqn:integral_changeofvari_hy}) can finally be expressed in terms of Lauricella function as (\ref{eqn:varphi_0}).

\section{LTAT of HARQ assisted OMA Scheme}\label{app:ltat_tdma}
With (\ref{eqn:noma_harq_throughput}), the LTAT of HARQ assisted OMA scheme is expressed as 
\begin{equation}\label{eqn:throughput_tdma}
\eta_{\rm{OMA}}  = \frac{{{R_1}\left( {1 - {\mathcal O_{{\rm{OMA}},K,{o_1}}}} \right) + {R_2}\left( {1 - {\mathcal O_{{\rm{OMA}},K,{o_2}}}} \right)}}{1+{\sum\nolimits_{\kappa  = 1}^{K - 1} {\left( {{\mathcal O_{{\rm{OMA}},\kappa ,{o_1}}} + {\mathcal O_{{\rm{OMA}},\kappa ,{o_2}}} - {\mathcal O_{{\rm{OMA}},\kappa ,{o_1},{o_2}}}} \right)} }},
\end{equation}
where ${\mathcal O_{{\rm{OMA}},K,{o_1}}}$, ${\mathcal O_{{\rm{OMA}},K,{o_2}}}$ and ${\mathcal O_{{\rm OMA},K,{o_1},{o_2}}}$ are respectively given by
\begin{align}\label{eqn:outage_1_tdma}
{\mathcal O_{{\rm{OMA}},K,{o_1}}} &= \mathbb P\left[ {{{\bar \Xi }_{{o_1}}}} \right]  = \mathbb P\left[ {{{\bar \Xi }_{{o_1}}},\left( {\bigcup\limits_{k = 1}^K {{\Xi _{{o_2},k}}} } \right)\bigcup {{{\bar \Xi }_{{o_2}}}} } \right]= \sum\limits_{k = 1}^K {\mathbb P\left[ {{{\bar \Xi }_{{o_1}}},{\Xi _{{o_2},k}}} \right]}  + \mathbb P\left[ {{{\bar \Xi }_{{o_1}}},{{\bar \Xi }_{{o_2}}}} \right],
\end{align}

\begin{align}\label{eqn:outage_2_tdma}
{\mathcal O_{{\rm{OMA}},K,{o_2}}} &= \mathbb P[{{\bar \Xi }_{{o_2}}}] = \mathbb P\left[ {\left( {\bigcup\limits_{l = 1}^K {{\Xi _{{o_1},l}}} } \right)\bigcup {{{\bar \Xi }_{{o_1}}}} ,{{\bar \Xi }_{{o_2}}}} \right] = \sum\limits_{l = 1}^K {\mathbb P\left[{{\Xi _{{o_1},l}} ,{{\bar \Xi }_{{o_2}}}} \right]}  + \mathbb P\left[ {{{\bar \Xi }_{{o_1}}},{{\bar \Xi }_{{o_2}}}} \right],
\end{align}
\begin{equation}\label{eqn:out12_tdma}
{\mathcal O_{{\rm OMA},K,{o_1},{o_2}}} = \mathbb P\left[ {{{\bar \Xi }_{{o_1}}},{{\bar \Xi }_{{o_2}}}} \right].
\end{equation}
Herein, ${{\bar \Xi }_{{o_i}}} $ denotes the outage event at user $i$ after $K$ HARQ rounds and ${\Xi _{{o_i},k}}$ represents the successful decoding event at user $i$ after $k$ HARQ rounds. To proceed, $\mathbb P\left[ {{{\bar \Xi }_{{o_1}}},{\Xi _{{o_2},k}}} \right]$, $\mathbb P\left[ {{\Xi _{{o_1},l}},{{\bar \Xi }_{{o_2}}}} \right]$ and $ \mathbb P\left[ {{{\bar \Xi }_{{o_1}}},{{\bar \Xi }_{{o_2}}}} \right]$ will be derived one by one.

From information-theoretical perspective, $\mathbb P\left[ {{{\bar \Xi }_{{o_1}}},{\Xi _{{o_2},k}}} \right]$ is given by 
\begin{equation}\label{eqn:tdma_out11}
\mathbb P\left[ {{{\bar \Xi }_{{o_1}}},{\Xi _{{o_2},k}}} \right]
 = \mathbb P\left[ \begin{array}{l}
\bigcap\limits_{j = 1}^k {{\beta ^2}\mathcal I\left( {\frac{{P{{\left| {{h_{z{o_1},j}}} \right|}^2}\ell \left( {{d_1}} \right)}}{{{I_{1,j}} + {\sigma ^2}}}} \right) < {R_1}} ,\bigcap\limits_{j = k + 1}^K {\mathcal I\left( {\frac{{P{{\left| {{h_{z{o_1},j}}} \right|}^2}\ell \left( {{d_1}} \right)}}{{{I_{1,j}} + {\sigma ^2}}}} \right) < {R_1},} \\
\bigcap\limits_{j = 1}^{k - 1} {\left( {1 - {\beta ^2}} \right)\mathcal I\left( {\frac{{P{{\left| {{h_{z{o_2},j}}} \right|}^2}\ell \left( {{d_2}} \right)}}{{{I_{2,j}} + {\sigma ^2}}}} \right) < {R_2}} ,\left( {1 - {\beta ^2}} \right)\mathcal I\left( {\frac{{P{{\left| {{h_{z{o_2},k}}} \right|}^2}\ell \left( {{d_2}} \right)}}{{{I_{2,k}} + {\sigma ^2}}}} \right) \ge {R_2}
\end{array} \right].
\end{equation}
Similar to Appendix \ref{app:outage_o1_1_f}, $\mathbb P\left( {{{\bar \Xi }_{{o_1}}},{\Xi _{{o_2},k}}} \right)$ can be derived as
\begin{align}\label{eqn:tdma_out11_fin}
&\mathbb P\left[ {{{\bar \Xi }_{{o_1}}},{\Xi _{{o_2},k}} } \right] =\sum\limits_{{\tau _1} = 0}^k {\sum\limits_{{\tau _2} = 0}^{K - k} {\sum\limits_{{\tau _3} = 0}^{k - 1} {{{\left( { - 1} \right)}^{\sum\limits_{j = 1}^3 {{\tau _j}} }}C_k^{{\tau _1}}C_{K - k}^{{\tau _2}}C_{k - 1}^{{\tau _3}}} \times} }  \notag\\
 & \Psi \left( {\left( {\frac{{{2^{\frac{{{R_1}}}{{{\beta ^2}}}}} - 1}}{{\ell \left( {{d_1}} \right)}},\frac{{{2^{{R_1}}} - 1}}{{\ell \left( {{d_1}} \right)}}} \right),\left( {{\tau _1},{\tau _2}} \right);\frac{{{2^{\frac{{{R_2}}}{{1 - {\beta ^2}}}}} - 1}}{{\ell \left( {{d_2}} \right)}},{\tau _3} + 1} \right).
\end{align}

With the same approach, we can prove
\begin{align}\label{eqn:tdma_outlo2}
&\mathbb P\left[ {{\Xi _{{o_1},l}},{{\bar \Xi }_{{o_2}}}} \right] = \sum\limits_{{\tau _1} = 0}^{l - 1} {\sum\limits_{{\tau _2} = 0}^l {\sum\limits_{{\tau _3} = 0}^{K - l} {{{\left( { - 1} \right)}^{\sum\limits_{j = 1}^3 {{\tau _j}} }}C_{l - 1}^{{\tau _1}}C_l^{{\tau _2}}C_{K - l}^{{\tau _3}}} \times} } \notag\\
&  \Psi \left( {\frac{{{2^{\frac{{{R_1}}}{{{\beta ^2}}}}} - 1}}{{\ell \left( {{d_1}} \right)}},{\tau _1} + 1;\left( {\frac{{{2^{\frac{{{R_2}}}{{1 - {\beta ^2}}}}} - 1}}{{\ell \left( {{d_2}} \right)}},\frac{{{2^{{R_2}}} - 1}}{{\ell \left( {{d_2}} \right)}}} \right),\left( {{\tau _2},{\tau _3}} \right)} \right).
\end{align}

Analogously, $\mathbb P\left( {{{\bar \Xi }_{{o_1}}},{{\bar \Xi }_{{o_2}}}} \right)$ follows as
\begin{equation}\label{eqn:tdma_out12_def}
 \mathbb P\left[ {{{\bar \Xi }_{{o_1}}},{{\bar \Xi }_{{o_2}}}} \right] = \mathbb P\left[\begin{array}{l}
\bigcap\limits_{j = 1}^K {\left( {{\beta ^2}\mathcal I\left( {\frac{{P{{\left| {{h_{z{o_1},j}}} \right|}^2}\ell \left( {{d_1}} \right)}}{{{I_{1,j}} + {\sigma ^2}}}} \right) < {R_1}} \right)} ,
\bigcap\limits_{j = 1}^K {\left( {\left( {1 - {\beta ^2}} \right)\mathcal I\left( {\frac{{P{{\left| {{h_{z{o_2},j}}} \right|}^2}\ell \left( {{d_2}} \right)}}{{{I_{2,j}} + {\sigma ^2}}}} \right) < {R_2}} \right)}
\end{array} \right].
\end{equation}
After some algebraic manipulations, it follows that
\begin{equation}\label{eqn:tdma_out12_deffina}
 \mathbb P\left[ {{{\bar \Xi }_{{o_1}}},{{\bar \Xi }_{{o_2}}}} \right] = \sum\limits_{{\tau _1} = 0}^K {\sum\limits_{{\tau _2} = 0}^K {{{\left( { - 1} \right)}^{\sum\limits_{j = 1}^2 {{\tau _j}} }}C_K^{{\tau _1}}C_K^{{\tau _2}}
 } } \Psi \left( {\frac{{{2^{\frac{{{R_1}}}{{{\beta ^2}}}}} - 1}}{{\ell \left( {{d_1}} \right)}},{\tau _1};\frac{{{2^{\frac{{{R_2}}}{{1 - {\beta ^2}}}}} - 1}}{{\ell \left( {{d_2}} \right)}},{\tau _2}} \right).
\end{equation}

\section{Proof of Remark \ref{rem:joint_power}}\label{app:joint_power}
Similar to (\ref{eqn:opt_prob_simp}), the problem of joint rate and power optimization for the HARQ assisted OMA scheme can be formulated as
\begin{equation}\label{eqn:opt_prob_simptdma}
\begin{array}{*{20}{cl}}
{\mathop {\rm maximize}\limits_{R_1, R_2,\beta^2} }&{\eta_{\rm{OMA}}}\\
{{\rm{subject}}\,{\rm{to}}}&{\mathcal O_{{\rm{OMA}},K,o_i} \le \varepsilon_i},\, i= 1,2\\
\end{array}
\end{equation}
Herein, we consider the case of $K=1$ and $\varepsilon_1=\varepsilon_2=\varepsilon$. According to (\ref{eqn:throughput_tdma}), (\ref{eqn:outage_1_tdma}) and (\ref{eqn:outage_2_tdma}), the LTAT and outage probabilities of the OMA scheme with $K=1$ are
\begin{align}\label{eqn:ltat_oma_k1}
 \eta _{{\rm{OMA}}}^{K = 1} &= {R_1}{e^{ - \frac{{{\sigma ^2}}}{P}\frac{{{2^{\frac{{{R_1}}}{{{\beta ^2}}}}} - 1}}{{\ell \left( {{d_1}} \right)}} - \pi \lambda {{\left( {\frac{{{2^{\frac{{{R_1}}}{{{\beta ^2}}}}} - 1}}{{\ell \left( {{d_1}} \right)}}} \right)}^{\frac{2}{\alpha }}}{\rm{B}}\left( { - \frac{2}{\alpha } + 1,\frac{2}{\alpha } + 1} \right)}}+ {R_2}{e^{ - \frac{{{\sigma ^2}}}{P}\frac{{{2^{\frac{{{R_2}}}{{1 - {\beta ^2}}}}} - 1}}{{\ell \left( {{d_2}} \right)}} - \pi \lambda {{\left( {\frac{{{2^{\frac{{{R_2}}}{{1 - {\beta ^2}}}}} - 1}}{{\ell \left( {{d_2}} \right)}}} \right)}^{\frac{2}{\alpha }}}{\rm{B}}\left( { - \frac{2}{\alpha } + 1,\frac{2}{\alpha } + 1} \right)}}\notag\\
 &= {\beta ^2}{\phi _1}\left( {\frac{{{R_1}}}{{{\beta ^2}}}} \right) + \left( {1 - {\beta ^2}} \right){\phi _2}\left( {\frac{{{R_2}}}{{1 - {\beta ^2}}}} \right),
\end{align}
\begin{equation}\label{eqn:out_tdma_k11}
{\mathcal O_{{\rm{OMA}},K,o_1}^{K=1}} = 1-{\vartheta _1}\left( \frac{R_1}{\beta^2} \right),
\end{equation}
\begin{equation}\label{eqn:out_tdma_k12}
{\mathcal O_{{\rm{OMA}},K,o_2}^{K=1}} = 1-{\vartheta _2}\left( \frac{R_2}{1-\beta^2} \right),
\end{equation}
where ${\phi _i}\left( x \right) = x{e^{ - \frac{{{\sigma ^2}}}{P}\frac{{{2^x} - 1}}{{\ell \left( {{d_i}} \right)}}{ - \pi \lambda {{\left( {\frac{{{2^x} - 1}}{{\ell \left( {{d_i}} \right)}}} \right)}^{\frac{2}{\alpha }}}{\rm{B}}\left( { - \frac{2}{\alpha } + 1,\frac{2}{\alpha } + 1} \right)}}}$ and ${\vartheta _i}\left( x \right) = {e^{ - \frac{{{\sigma ^2}}}{P}\frac{{{2^x} - 1}}{{\ell \left( {{d_i}} \right)}}{ - \pi \lambda {{\left( {\frac{{{2^x} - 1}}{{\ell \left( {{d_i}} \right)}}} \right)}^{\frac{2}{\alpha }}}{\rm{B}}\left( { - \frac{2}{\alpha } + 1,\frac{2}{\alpha } + 1} \right)}}}$. Instead of jointly optimizing $R_1$, $R_2$ and ${\beta}^2$, we introduce $z_1 = {\frac{{{R_1}}}{{{\beta ^2}}}}$, $z_2 = {\frac{{{R_2}}}{{1 - {\beta ^2}}}}$, and the optimization problem (\ref{eqn:opt_prob_simptdma}) can be reformulated by using (\ref{eqn:ltat_oma_k1}), (\ref{eqn:out_tdma_k11}) and (\ref{eqn:out_tdma_k12}) as
\begin{equation}\label{eqn:opt_prob_simp0}
\begin{array}{*{20}{cl}}
{\mathop {\rm maximize}\limits_{z_1, z_2, \beta^2} }&{{\beta ^2}{\phi _1}\left( {{z_1}} \right) + \left( {1 - {\beta ^2}} \right){\phi _2}\left( {{z_2}} \right)}\\
{{\rm{subject}}\,{\rm{to}}}&{{\vartheta _i}\left( z_i \right) \ge 1 - \varepsilon },\, i= 1,2,
\end{array}
\end{equation}
With decomposition theory \cite{palomar2006tutorial}, the optimization with respect to $z_1$ and $z_2$ can be decoupled as
\begin{equation}\label{eqn:opt_prob_simp1}
\begin{array}{*{20}{cl}}
{\mathop {\rm maximize}\limits_{z_i} }&{{\phi _i}\left( {{z_i}} \right) }\\
{{\rm{subject}}\,{\rm{to}}}&{{\vartheta _i}\left( z_i \right) \ge 1 - \varepsilon },
\end{array}
\end{equation}
Noticing that $d_1 < d_2$, it follows that ${\phi _1}({z_1}^*) > {\phi _2}({z_2}^*)$. After obtaining the optimal ${z_1}^*$ and ${z_2}^*$, it is not hard to prove that the optimal LTAT is an increasing function of $\beta^2$, i.e., ${\beta ^2}\left( {{\phi _1}\left( {{z_1}^*} \right) - {\phi _2}\left( {{z_2}^*} \right)} \right) + {\phi _2}\left( {{z_2}^*} \right)$. The maximal LTAT is achieved if and only if ${\beta^*}^2=1$. Hence, we have ${R_2} = \left( {1 - {\beta ^2}} \right){z_2}^* = 0$bps/Hz. The proof is then completed.

\bibliographystyle{ieeetran}
\bibliography{manuscript}

\end{document}